\documentclass[a4paper,11pt,twoside,onecolumn]{amsart}%
\usepackage{amsfonts, amssymb, graphicx}

\newcommand{\comment}[1]{}

\newtheorem{theorem}{Theorem}

\newtheorem{claim}[theorem]{Claim}
\newcommand{\claimproof}[2]%
{\noindent{\em Proof of Claim \ref{#1}.}
#2\hspace*{\fill}$\Box$~~~~~\vspace{5mm} }

\newtheorem{corollary}[theorem]{Corollary}

\newtheorem{definition}[theorem]{Definition}

\newtheorem{lemma}[theorem]{Lemma}

\newcommand{\FF}{\mathbb{F}}
\newcommand{\CC}{\mathbb{C}}
\newcommand{\RR}{\mathbb{R}}
\newcommand{\PP}{\mathbb{P}}
\newcommand{\QQ}{\mathbb{Q}}
\newcommand{\ZZ}{\mathbb{Z}}
\newcommand{\NN}{\mathbb{N}}

\newcommand{\cU}{\mathcal{U}}

\newcommand{\cI}{\mathcal{I}}
\newcommand{\cH}{\mathcal{H}}

\newcommand{\cP}{\mathcal{P}}
\newcommand{\cC}{\mathcal{C}}
\newcommand{\bcP}{\mathfrak{P}}

\newcommand{\sps}{\Sigma\Pi\Sigma}
\newcommand{\rk}{\text{rk}}
\newcommand{\fanrk}{\text{ind-fanin}}
\newcommand{\simp}{\text{sim}}
\newcommand{\degr}{\text{deg}}
\newcommand{\lrsp}{\text{sp}}
\newcommand{\scal}{\text{sc}}
\newcommand{\piv}{\text{rep}}
\newcommand{\nod}{\text{nod}}
\renewcommand{\mod}{\text{mod~}}
\newcommand{\poly}{\text{poly}}
\newcommand{\matnucleus}{\text{mat-nucleus}}

\newcommand{\ideal}[1]{\langle {#1} \rangle}
\newcommand{\radsp}{\text{radsp}}
\newcommand{\cl}{\text{cl}}
\newcommand{\fami}{\text{fam}}
\newcommand{\trun}{\text{trun}}
\newcommand{\pa}{\text{Part}}
\newcommand{\sg}{\text{SG}}
\newcommand{\EX}{\hbox{\bf E}}
\newcommand{\pr}{{\rm Pr}}
\newcommand{\cE}{\mathcal{E}}

\newcommand{\sgkb}[2]{SG_k(#1 #2)}

\def\vec#1{\overline{#1}}

\def\({\left(}
\def\){\right)}
\def\<{\langle}
\def\>{\rangle}

\def\le{\leqslant}
\def\ge{\geqslant}

\begin{document}

\author{Nitin Saxena}
\thanks{Hausdorff Center for Mathematics, Bonn - 53115, Germany. {\tt ns@hcm.uni-bonn.de}}

\author{C. Seshadhri}
\thanks{IBM Almaden, San Jose - 95126, USA. {\tt csesha@gmail.com}}

\title[SG Configurations \& Rank Bounds]{From Sylvester-Gallai Configurations to Rank Bounds: 
Improved Black-box Identity Test for Depth-3 Circuits}

\begin{abstract}
We study the problem of identity testing for depth-$3$ circuits
of top fanin $k$ and degree $d$ (called $\sps(k,d)$ identities). We 
give a new structure theorem for such identities.
A direct application of our theorem improves the known 
deterministic $d^{k^{O(k)}}$-time black-box identity test over rationals (Kayal \& Saraf, FOCS 2009) 
to one that takes $d^{O(k^2)}$-time.
Our structure theorem
essentially says that the number of independent variables in a real depth-$3$ identity
is very small. This theorem settles 
affirmatively the stronger rank conjectures posed by Dvir \& Shpilka (STOC 2005) 
and Kayal \& Saraf (FOCS 2009). 
Our techniques provide a unified framework that actually beats all 
known rank bounds and hence gives the best running time (for \emph{every field})
for black-box identity tests.

Our main theorem (almost optimally) pins down the relation between higher dimensional Sylvester-Gallai
theorems and the rank of depth-$3$ identities in a very transparent
manner. The existence of this was hinted at by Dvir \& Shpilka (STOC 2005), but
first proven, for reals, by Kayal \& Saraf (FOCS 2009). 
We introduce the concept of \emph{Sylvester-Gallai rank bounds} for any field,
and show the intimate connection between this and depth-$3$ identity rank bounds.
We also prove the first ever theorem about high dimensional Sylvester-Gallai
configurations over \emph{any field}. Our proofs and techniques are very
different from previous results and devise a very interesting ensemble of combinatorics and algebra.
The latter concepts are ideal theoretic and involve a new Chinese remainder theorem.
Our proof methods explain the structure of \emph{any} depth-$3$
identity $C$: there is a \emph{nucleus} of $C$ that forms a low rank identity, while the remainder 
is a high dimensional Sylvester-Gallai configuration.   

\end{abstract}

\date{}	
\maketitle

\section{Introduction}

Polynomial identity testing (PIT) ranks as one of the most important open problems 
in the intersection of algebra and computer science. 
We are provided an arithmetic circuit that computes a polynomial $p(x_1,x_2,\cdots,x_n)$
over a field $\FF$, and we wish to test if $p$ is identically zero
(in other words, if $p$ is the zero polynomial). 
In the black-box setting, the circuit
is provided as a black-box and we are only allowed to evaluate the polynomial
$p$ at various domain points. 
The main goal is to devise a deterministic polynomial time
algorithm for PIT.
Kabanets \& Impagliazzo~\cite{KI04} and Agrawal~\cite{A05, A06}
have shown connections between deterministic algorithms for identity testing and circuit lower bounds,
emphasizing the importance of this problem. To know more about the current state of the general identity 
testing problem see the surveys \cite{S09, AS09}.

The first randomized polynomial time PIT algorithm, 
which was a black-box algorithm,
was given (independently) by Schwartz~\cite{Sch80} and Zippel~\cite{Z79}.
Randomized algorithms that use less
randomness were given by Chen \& Kao~\cite{CK00}, Lewin \& Vadhan~\cite{LV98},
and Agrawal \& Biswas~\cite{AB03}.
Klivans \& Spielman~\cite{ks01} observed that 
even for depth-$3$ circuits for bounded top fanin,
deterministic identity testing was open. Progress
towards this was first made by Dvir \& Shpilka~\cite{DS06},
who gave a quasi-polynomial time algorithm, although with
a doubly-exponential dependence on the top fanin.
The problem was resolved by a polynomial time algorithm given by
Kayal and Saxena~\cite{KS07}, with a running time exponential
in the top fanin. As expected, the current understanding of depth-$4$ circuits is even 
more sparse. Identity tests are known only for rather special depth-$4$ circuits 
\cite{AM07, S08, SV09, KMSV09}. 
Why is progress restricted to such small depth circuits?
Agrawal and Vinay~\cite{AV08} showed that
an efficient black-box identity test for depth-$4$
circuits will actually give a quasi-polynomial black-box test, and subexponential lower bounds, 
for circuits of \emph{all depths} (that compute {\em low degree} polynomials).
Thus, understanding depth-$3$ identities seems to be a natural first step towards
the goal of proving more general lower bounds. 

For deterministic black-box testing, the first results
were given by Karnin \& Shpilka \cite{KSh08}. 
Based on results in~\cite{DS06},
they gave an algorithm for bounded top fanin depth-$3$ circuits having a quasi-polynomial
running time (with a doubly-exponential dependence
on the top fanin). The dependence on the top fanin was later improved (to singly-exponential) 
by the rank bound results of Saxena \& Seshadhri \cite{SS09} (for {\em any} $\FF$). But the
time complexity also had a quasi-polynomial dependence on the degree of the circuit. This 
dependence is inevitable in rank-based methods over finite fields (as shown by \cite{KS07}).
However, over the field of rationals, Kayal \& Saraf \cite{KSar09} showed how to remove this 
quasi-polynomial dependence on the degree at the cost of doubly-exponential dependence on
the top fanin, thus giving a polynomial time complexity for bounded top fanin.  
In this work we achieve the best of the two works \cite{SS09} and \cite{KSar09}, i.e. we
prove (for rationals) a time complexity that depends only {\em polynomially} on the degree and ``only'' 
{\em singly}-exponentially on the fanin. 

In a quite striking result, Kayal \& Saraf~\cite{KSar09} proved how Sylvester-Gallai theorems
can get better rank bounds over the reals.
We introduce the concept of \emph{Sylvester-Gallai rank bounds} that deals
with the rank of vectors (over some given field) that have some special
incidence properties. 
This is a very convenient way to express known Sylvester-Gallai results.
These are inspired by the famous
Sylvester-Gallai theorem about point-line incidences.
We show how this very interesting quantity is tightly connected
to depth-$3$ identities. Sylvester-Gallai rank bounds 
over high dimensions were known over the reals, and 
are used to prove depth-$3$ rank bounds over reals.
We prove the first ever theorem for high 
dimensional Sylvester-Gallai configurations 
over \emph{any field}.


\subsection{Definitions and Previous Work}

This work focuses on depth-$3$ circuits. A structural study of 
depth-$3$ identities was initiated in \cite{DS06} by defining a notion of {\em rank} 
of {\em simple} and {\em 
minimal} identities. A depth-$3$ circuit $C$ over a field $\FF$ is: 
$$C(x_1,\ldots,x_n) = \sum_{i=1}^k T_i$$ 
where, $T_i$ ({\em a multiplication term}) is a product of $d_i$ linear polynomials 
$\ell_{i,j}$ over $\FF$. Note that for the purposes of studying identities we can 
assume wlog (by {\em homogenization}) that $\ell_{i,j}$'s are linear {\em forms} (i.e. 
linear polynomials 
with a zero constant coefficient) and that $d_1=\cdots=d_k=:d$. 
Such a circuit is referred to as a $\sps(k,d)$ circuit (or $\sps(k,d,n)$ depending on the context),
where $k$ is the \emph{top fanin} of $C$ and $d$ is the {\em degree} of $C$. 
We give a few definitions from~\cite{DS06}.

\begin{definition} {\bf[Simple Circuit]}
$C$ is a {\em simple} circuit if there is no nonzero linear form dividing all 
the $T_i$'s. 

{\bf[Minimal Circuit]} $C$ is a {\em minimal} circuit if for every proper subset 
$S\subset[k]$, $\sum_{i\in S}T_i$ is nonzero. 

{\bf[Rank of a circuit]} Every $\ell_{i,j}$ can be seen as an $n$-dimensional vector over $\FF$.
The {\em rank} of the circuit, $\rk(C)$, is defined as 
the rank of the set of all linear forms $\ell_{i,j}$'s viewed as $n$-dimensional vectors.
\end{definition}

Can all the forms $\ell_{i,j}$ be independent, or must there
be relations between them?
The rank can be interpreted as the minimum number of variables
that are required to express $C$. There exists a linear
transformation converting the $n$ variables of the circuit
into $rank(C)$ independent variables. A trivial upper bound on the rank (for any $\sps$-circuit) 
is $kd$, since that is the total number of
linear forms involved in $C$. 
The rank is a fundamental property of a $\sps(k,d)$ circuit
and it is crucial to understand how large this can be 
for identities.
A substantially smaller rank bound than $kd$ shows that 
identities do not have as many ``degrees of freedom" as general
circuits, and leads to deterministic identity tests. 
Furthermore, the techniques used to prove rank bounds 
show us structural properties of identities that may 
suggest directions to resolve PIT for $\sps(k,d)$ circuits.

The rank bounds, in addition to being a natural property of identities, have found 
applications in black-box identity testing \cite{KSh08} and learning $\sps$ circuits
\cite{Sh09, KSh09}. The result of \cite{KSh08} showed rank bounds imply
black-box testers: if $R(\FF,k,d)$ is a rank bound for
simple minimal $\sps(k,d,n)$ identities over field $\FF$, then there is a deterministic 
black-box identity tester for such circuits, that runs in $\poly(n,d^{R(\FF,k,d)})$ 
$\FF$-operations. (For the time complexity over $\QQ$, we actually count the {\em bit 
operations}.)

Dvir \& Shpilka~\cite{DS06} proved that the
rank of a simple, minimal $\sps(k,d)$ identity is bounded by $2^{O(k^2)}(\log d)^{k-2}$.  
This rank bound was improved to $O(k^3\log d)$ by Saxena \& Seshadhri \cite{SS09}.
Fairly basic identity constructions show that
the rank is $\Omega(k)$ over the reals and $\Omega(k \log d)$ for
finite fields~\cite{DS06,KS07,SS09}.
Dvir \& Shpilka \cite{DS06} conjectured that $\rk(C)$ should be some $\poly(k)$ over the reals. 
Through a very insightful
use of Sylvester-Gallai theorems, Kayal \& Saraf \cite{KSar09} subsequently bounded
the rank of identities, over reals, by $k^{O(k)}$. This means that for a constant top fanin 
circuit,
the rank of identities is constant, independent of the degree. This also
leads to the first truly polynomial-time deterministic black-box identity testers for this case.

Unfortunately, as soon as $k$ becomes even $\Omega(\log n)$, this bound
becomes trivial. We improve this rank bound 
exponentially, to $O(k^2)$, which is almost optimal. 
This gives a major improvement
in the running time of the black-box testers.
We also improve the rank bounds for general fields
from $O(k^3\log d)$ to $O(k^2\log d)$. We emphasize that
we give a unified framework to prove all these results.
Table~\ref{tab} should make it easier to compare the various
bounds.


\begin{table}[h] 
\caption{Known rank bounds and black-box PIT}
\centering 
\begin{tabular}{l l l l l} 
\hline\hline 
Paper & Result & Asymptotics & over field \\ [0.5ex]
\hline 
& rank bound & $2^{k^2}\log^{k-2}d$ & any \\[-0.6ex]
\raisebox{1.5ex}{\cite{DS06, KSh08}} & time complexity & $n d^{(2^{k^2}\log^{k-2}d)}$ & any \\[1ex]
& rank bound & $k^3\log d$ & any \\[-1ex]
\raisebox{1.5ex}{\cite{SS09}} & time complexity & $n d^{k^3\log d}$ & any \\[1ex]
& rank bound & $k^k$ & $\RR$ \\[-1ex]
\raisebox{1.5ex}{\cite{KSar09}} & time complexity & $n d^{(k^k)}$ & $\QQ$ \\[1ex]
& rank bound & $k^2$ & $\RR$ \\[-1ex]
\raisebox{1.5ex}{Ours} & time complexity & $n d^{k^2}$ & $\QQ$ \\[1ex]
& rank bound & $k^2\log d$ & any \\
& time complexity & $n d^{k^2\log d}$ & any \\[1ex]
\hline 
\end{tabular}
\label{tab}
\end{table}

Kayal \& Saraf~\cite{KSar09} connect Sylvester-Gallai theorems to rank bounds. They need
advanced versions of these theorems that deal with colored points and have to prove
certain \emph{hyperplane decomposition theorems}. We make the connection much more
transparent (at the loss of some color from the theorems). 
We reiterate that our techniques are completely different, and employ
a very powerful algebraic framework to dissect identities.
This allows us to use as a ``black-box'' the most basic form 
of the higher dimensional Sylvester-Gallai theorems. 

\subsection{Our Results}

Before we state our results, it will be helpful to understand
Sylvester-Gallai configurations.
A set of points $S$ with the property 
that every line through two points of $S$ passes through
a third point in $S$ is called a \emph{Sylvester-Gallai configuration}.
The famous Sylvester-Gallai theorem states: for a set $S$ of points
in ${\RR}^2$, not all collinear, there exists
a line passing through exactly two points of $S$. 
In other words, the only Sylvester-Gallai configuration in ${\RR}^2$
is a set of collinear points.
This basic
theorem about point-line incidences was extended to 
higher dimensions~\cite{Han65,BE67}.
We introduce the notion of \emph{Sylvester-Gallai rank bounds}.
This is a clean and convenient way of expressing these theorems.

\begin{definition}
Let $S$ be a finite subset of the {\em projective space} $\FF\PP^{n}$.
Alternately, $S$ is a subset of vectors in $\FF^{n+1}$
without \emph{multiples}: no two vectors in $S$ are scalar
multiples of each other\footnote{When $|\FF|>|S|$, such an $S$ is, wlog, a subset of
distinct vectors with first coordinate $1$.}.
Suppose, for every set $V \subset S$ of $k$ linearly independent vectors, the linear
span of $V$ contains at least $k+1$ vectors of $S$. 
Then, the set $S$ is said to be \emph{$\sg_k$-closed}.

The largest possible rank of an $\sg_k$-closed set of at most $m$ vectors in $\FF^n$ (for any $n$)
is denoted by $\sg_k(\FF,m)$.
\end{definition}

The classic Sylvester-Gallai theorem essentially 
states\footnote{To see this, take an $SG_2$-closed set $S$ of vectors. Think of each
vector being represented by an infinite line through the origin, hence giving
a set $S$ in the projective space. Take a $2$-dimensional plane $P$ not passing through
the origin and take the set of intersection points $I$ of the lines in $S$ with $P$.
Observe that the coplanar points $I$ have the property that a line
passing through two points of $I$ passes through a third point of $I$.} that for all $m$, 
$\sg_2(\RR,m)\le2$.
Higher dimensional analogues~\cite{Han65,BE67} prove that $\sg_k(\RR,m) \leq 2(k-1)$.
One of our auxiliary theorems is such a statement for \emph{all fields}.

\begin{theorem}[$\sg_k$ for all fields] \label{thm-sgk} For any field $\FF$ and 
$k,m\in\NN^{>1}$, $\sg_k(\FF,m)\leq 9k\lg m$.
\end{theorem}

Our main theorem is a simple, clean expression of how Sylvester-Gallai
influences identities.

\begin{theorem}[From $\sg_k$ to Rank]
Let $|\FF|>d$. The rank of 
a simple and minimal $\sps(k,d)$ identity over $\FF$ is at most $2k^2 + k\cdot\sg_k(\FF,d)$.
\end{theorem}	

\noindent {\bf Remark.}
If $\FF$ is small, then we choose an extension $\FF'\supset\FF$ of size $> d$ and get a
rank bound with $\sg_k(\FF',d)$. \\

Plugging in $\sg_k$-rank bounds gives us the desired theorem for depth-$3$ identities.
We have a slightly stronger version of the above theorem that we use to
get better constants (refer to Theorem~\ref{thm-rank}). 

\begin{theorem}[Depth-$3$ Rank Bounds]\label{thm-main}
Let $C$ be a $\sps(k,d)$ circuit, over field $\FF$, that is simple, minimal and zero. Then,
\begin{itemize}
	\item For $\FF = \RR$, $\rk(C)<3k^2$. 
	\item For any $\FF$, $\rk(C)<3k^2(\lg 2d)$.
\end{itemize}
\end{theorem}

As discussed before, a direct application of this result to Lemma 4.10 of \cite{KSh08}
gives a deterministic black-box identity test for $\sps(k,d,n)$ circuits (we
will only discuss $\QQ$ here as the other statement is analogous). 
Formally, we get the following {\em hitting set generator} for $\sps$ circuits with real
coefficients.

\begin{corollary}[Black-box PIT over $\QQ$]
There is a deterministic algorithm that takes as input a triple
$(k,d,n)$ of natural numbers and in time $\poly(nd^{k^2})$,
outputs a hitting set $\cH\subset\ZZ^n$ with the following properties:
\begin{itemize}
\item[1)] Any $\sps(k,d,n)$ circuit $C$ over $\RR$
computes the zero polynomial iff $\forall a\in\cH$, $C(a) = 0$.
\item[2)] $\cH$ has at most $\poly(nd^{k^2})$ points.
\item[3)] The total bit-length of each point in $\cH$ is $\poly(kn\log d)$.
\end{itemize}
\end{corollary}
\noindent {\bf Remark.}
\begin{itemize}
\item[1)] 
Our black-box test 
has {\em quasi-polynomial} in $n$ time complexity (with polynomial-dependence on $d$) for top fanin as large as 
$k=\poly\log(n)$, 
and {\em sub-exponential} in $n$ time complexity (with polynomial-dependence on $d$) even for 
top fanin as large as $k=o(\sqrt{n})$. This is the first tester to achieve
such bounds.
\item[2)] The fact that the points in $\cH$ are integral and have ``small'' bit-length 
is important to estimate the time complexity of our algorithm in terms of {\em bit 
operations}. Thus, the hitting set generator takes at most $\poly(nd^{k^2})$ {\em bit} 
operations to compute $\cH$.
\end{itemize}

\section{Proof Outline, Ideas, and Organization}

Our proof of the rank bound comprises of several new ideas, both at the conceptual and the technical
levels. In this section we will give the basic intuition of the proof.
The three notions that are crucially used (or developed) in 
the proof are: ideal Chinese remaindering, matchings and Sylvester-Gallai rank bounds. These have
appeared (in some form) before in the works of Kayal \& Saxena \cite{KS07}, Saxena \& Seshadhri
\cite{SS09} and Kayal \& Saraf \cite{KSar09} respectively, to prove different kinds of results.
Here we use all three of them together to show quite a strong structure in $\sps$ identities. 
We will talk about them one by one in
the following three subsections outlining the three steps of the proof. Each step proves a new
property of identities which is interesting in its own right. The first two steps set
up the algebraic framework and prove theorems that hold for all fields.
The third step is where the Sylvester-Gallai theorems are brought in.
Some (new and crucial) algebraic lemmas and their proofs 
have been moved to the Appendix. The flow of the actual proof
will be identical to the overview that we now provide.

\subsection{Step 1: Matching the Gates in an Identity}

We will denote the set $\{1,\ldots,n\}$ by $[n]$. We fix 
the base field to be $\FF$, so the circuits
compute multivariate polynomials in the {\em polynomial ring} $R:=\FF[x_1,\ldots,x_n]$. 

A \emph{linear form} is a linear polynomial in $R$ with zero constant term. We will denote
the set of all linear forms by  
$L(R):= \left\{\sum_{i=1}^n a_ix_i \mid a_1,\ldots,a_n\in\FF\right\}$.
Clearly, $L(R)$ is a vector (or linear) space over $\FF$ and that will be quite useful.
Much of what we do shall deal with {\em multi}-sets of linear forms (sometimes polynomials in $R$ too), equivalence 
classes inside them, and various maps across them. A \emph{list} of linear
forms is a multi-set of forms with an arbitrary order
associated with them. The actual ordering is unimportant: 
we will heavily use maps between lists, and the ordering allows us to
define these maps unambiguously. The object, list, comes with all the usual set operations naturally defined.

\begin{definition} We collect some important definitions from \cite{SS09}:

{\bf[Multiplication term, $L(\cdot)$ \& $M(\cdot)$]} A {\em multiplication term} $f$ is an 
expression in $R$ given as (the product may have repeated $\ell$'s),
$f := c\cdot\prod_{\ell\in S}\ell$, where $c\in\FF^*$
and $S$ is a list of nonzero linear forms.
The {\em list of linear forms in $f$}, $L(f)$, is just the list $S$ of forms 
occurring in the product above. 
For a list $S$ of linear forms we define the {\em multiplication 
term of $S$}, $M(S)$, as $\prod_{\ell\in S}\ell$ or $1$ if $S=\phi$. 

{\bf[Forms in a Circuit]} We will represent a $\sps(k,d)$ circuit $C$ as a sum of $k$ multiplication terms
of degree $d$, $C = \sum_{i=1}^k T_i$. The list of {\em linear forms occurring in $C$} 
is $L(C):=$ $\bigcup_{i\in[k]}L(T_i)$. Note that $L(C)$ is a list of size exactly $kd$. 
The {\em rank of $C$}, $\rk(C)$, is just the number of linearly independent
linear forms in $L(C)$. (Remark: for the purposes of this paper $T_i$'s are given in circuit 
representation and thus the list $L(T_i)$ is unambiguously defined from $C$)

{\bf [Similar forms]} For any two polynomials $f,g\in R$ we call $f$ {\em similar to} $g$ if there 
exists $c\in\FF^*$ such that $f=cg$. We say $f$ {\em is similar to $g$ mod 
$I$}, for some ideal $I$ of $R$, if there is some $c\in\FF^*$ such that 
$f \equiv cg (\mod I)$. Note that ``similarity mod $I$" is an equivalence relation (reflexive, symmetric
and transitive) and partitions any list of polynomials into equivalence classes.

{\bf [Span $\lrsp(\cdot)$]} For any $S\subseteq L(R)$ we let $\lrsp(S)\subseteq L(R)$ be the 
{\em linear span} of the linear forms in $S$ over the field $\FF$. (Conventionally, 
$\lrsp(\emptyset)=\{0\}$.)

{\bf [Matchings]}
Let $U, V$ be lists of linear forms and $I$ be a subspace of $L(R)$. An {\em $I$-matching 
$\pi$ between $U, V$} is a bijection $\pi$ between lists $U, V$ such 
that: for all $\ell\in U$, $\pi(\ell)\in \FF^*\ell+I$.

When $f,g$ are multiplication terms, an {\em $I$-matching between $f,g$}
would mean an $I$-matching between $L(f), L(g)$.
\end{definition}
We will show that all the multiplication terms of a minimal $\sps$ identity can be matched by
a ``low'' rank space.
\begin{theorem}[Matching-Nucleus]\label{thm-mat-nucleus}
Let $C=T_1+\cdots+T_k$ be a $\sps(k,d)$ circuit that is minimal and zero. 
Then there exists a linear subspace $K$ of $L(R)$ such that:
\\\indent
1) $\rk(K)< k^2$.
\\\indent
2) $\forall i\in[k]$, there is a $K$-matching $\pi_i$ between $T_1, T_i$. 
\end{theorem}
The idea of matchings within identities was first introduced in~\cite{SS09},
but nothing as powerful as this theorem has been proven.
This theorem gives us a space of small rank, \emph{independent of $d$},
that contains most of the ``complexity" of $C$. All forms in $C$ outside $K$
are just mirrored in the various terms.
This starts connecting the algebra
of depth-$3$ identities to a combinatorial structure. Indeed, the graphical
picture (explained in detail below) that this theorem provides, really
gives an intuitive grasp on these identities. The proof
of this involves some interesting generalizations of the Chinese Remainder
Theorem to some special ideals.

\begin{definition}[mat-nucleus]
Let $C$ be a minimal $\sps(k,d)$ identity. The linear subspace $K$ given by Theorem
\ref{thm-mat-nucleus} is called {\em mat-nucleus of $C$}.
\end{definition}
The notion of mat-nucleus is easier to see in the following unusual representation of
the $\sps(4,d)$ circuit $C=\sum_{i\in[4]}T_i$. The four bubbles refer to the four 
multiplication terms of $C$ and the points inside the bubbles refer to the linear forms
in the terms. The proof of Theorem \ref{thm-mat-nucleus} gives mat-nucleus as the space 
generated by the linear forms in the dotted box. The linear forms that are not in 
mat-nucleus lie ``above" the mat-nucleus and are all (mat-nucleus)-matched, i.e. 
$\forall \ell\in (L(T_1)\setminus\matnucleus)$, there is a form similar to $\ell$ modulo 
$\matnucleus$ in each $(L(T_i)\setminus\matnucleus)$. Thus the essence of Theorem \ref{thm-mat-nucleus} 
is: the mat-nucleus part of the terms of $C$ has low rank $k^2$, while the part of the terms 
above mat-nucleus all look ``similar". 
\begin{center}\includegraphics[scale=0.43]{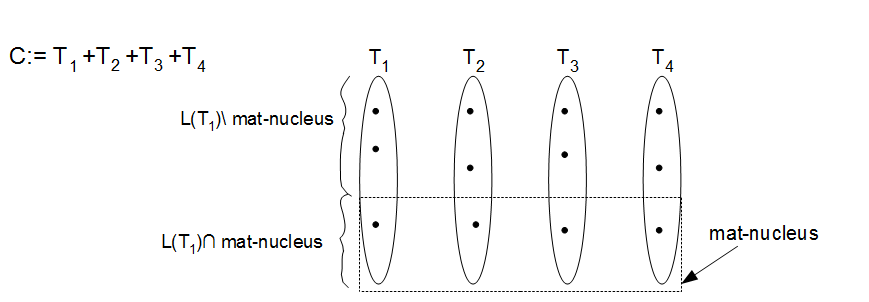}\end{center}

\subsubsection*{Proof Idea for Theorem \ref{thm-mat-nucleus}}

The key insight in the construction of mat-nucleus is a reinterpretation of the identity test 
of Kayal \& Saxena \cite{KS07} as a structural result for $\sps$ identities. Again, refer to
the following figure depicting a $\sps(4,d)$ circuit $C$ and think of each bubble having
$d$ linear forms. Roughly, \cite{KS07} showed that $C=0$ iff for every {\em path} 
$(v_1,v_2,v_3)$ (where $v_i\in L(T_i)$): $T_4\equiv0 (\mod v_1,v_2,v_3)$ or in ideal 
terms, $T_4\in\ideal{v_1,v_2,v_3}$. Thus, roughly, it is enough to go through all the $d^3$ paths
to certify the zeroness of $C$. This is why the time complexity of the identity test of 
\cite{KS07} is dominated by $d^k$.
\begin{center}\includegraphics[scale=0.35]{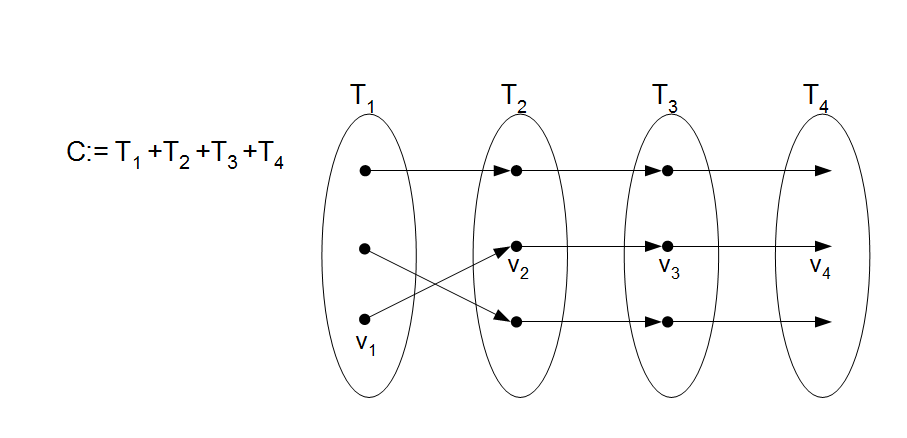}\end{center}
Now if we are given a $\sps(4,d)$ identity $C$ which is {\em minimal}, then we know that 
$T_1+T_2+T_3\ne0$. Thus, by applying the above interpretation of \cite{KS07} to $T_1+T_2+T_3$ 
we will get a path $(v_1,v_2)$ such that $T_3\notin\ideal{v_1,v_2}$. Since $C=0$ this means 
that $T_3+T_4\equiv0 (\mod v_1,v_2)$ but $T_3,T_4\not\equiv0 (\mod v_1,v_2)$ (if 
$T_4$ is in $\ideal{v_1,v_2}$ then so will be $T_3$). Thus, $T_3\equiv -T_4 (\mod v_1,v_2)$
is a nontrivial congruence and it immediately gives us a $\ideal{v_1,v_2}$-matching
between $T_3, T_4$ (see Lemma \ref{lem-I-match}). 
By repeating this argument with a different permutation of the terms we
could match different terms (by a different ideal), and finally we expect to match all the
terms (by the union of the various ideals).

This fantastic argument has numerous technical problems, but they can all be taken care
of by suitable algebraic generalizations. The main stumbling block is the presence
of \emph{repeating} forms. It could happen that $(\mod v_1)$, $v_2$ occurs in many
terms, or in the same term with a higher power.
The most important tool developed is an ideal version of Chinese
remaindering that forces us to consider not just linear forms $v_1, v_2$, but 
{\em multiplication terms} $v_1, v_2$ dividing $T_1, T_2$ respectively. We give the full 
proof in Section \ref{sec-mat-nucleus}.
(Interestingly, the {\em non}-blackbox identity test of
\cite{KS07} guides in devising a blackbox test of ``similar'' complexity over rationals.)

\subsection{Step 2: Certificate for Linear Independence of Gates}

Theorem \ref{thm-mat-nucleus} gives us a space $K$, of rank $<k^2$, that matches $T_1$ to each 
term $T_i$. In particular, this means that the list $L_K(T_i):=L(T_i)\cap K$ has the same cardinality
$d'$ for each $i\in[k]$. In fact, if we look at the corresponding
multiplication terms $K_i:=M(L_K(T_i))$, $i\in[k]$, then they again form a $\sps(k,d')$
identity! Precisely, $C'=\sum_{i\in[k]}\alpha_iK_i$ for some $\alpha_i$'s in $\FF^*$ (see Lemma
\ref{lem-nucleus}) is an identity. 
We would like $C'$ to somehow mimic the structure of $C$.
Of course $C'$ is simple but is it again minimal? Unfortunately, it may not be. 
For reasons that will be clear later,
minimality of $C'$ would have allowed us to go directly to Step 3. 
Now step 2 will involve increasing the space $K$ (but not by too much)
that gives us a $C'$ that ``behaves" like $C$.
Specifically,
if $T_1,\ldots,T_{k'}$ are {\em linearly independent} (i.e. $\nexists$ 
$\vec{\beta}\in\FF^{k'}\setminus\{\vec{0}\}$ s.t. $\sum_{i\in[k']}\beta_iT_i=0$), then so are
$K_1,\ldots,K_{k'}$. 

\begin{theorem}[Nucleus]\label{thm-nucleus} 
Let $C=\sum_{i\in[k]}T_i$ be a minimal $\sps(k,d)$ identity and let $\{T_i|i\in\cI\}$ be 
a maximal set of linearly
independent terms ($1\le k':=|\cI|<k$). Then there exists a linear subspace $K$ of $L(R)$ such that:
\begin{itemize}
\item[1)] $\rk(K)< 2k^2$.
\item[2)] $\forall i\in[k]$, there is a $K$-matching $\pi_i$ between $T_1, T_i$.
\item[3)] (Define $\forall i\in\cI$, $K_i:=M(L_K(T_i))$.) The terms $\{K_i|i\in\cI\}$ 
are linearly independent.
\end{itemize}
\end{theorem}

\begin{definition}[nucleus]
Let $C$ be a minimal $\sps(k,d)$ identity. The linear subspace $K$ given by Theorem
\ref{thm-nucleus} is called the {\em nucleus of $C$}. By Lemma~\ref{lem-nucleus}, the subspace $K$
induces an identity $C'=\sum_{i\in[k]} \alpha_i K_i$ which we call the {\em nucleus identity}.
\end{definition}

The notion of the nucleus is easier to grasp when $C$ is a $\sps(k,d)$ identity that is 
{\em strongly minimal}, i.e. $T_1,\ldots,T_{k-1}$ are linearly independent. Clearly, such a
$C$ is also minimal\footnote{If for some proper $S\subset[k]$, $\sum_{i\in S}T_i=$
$\sum_{i\in \vec{S}}T_i=0$ then linear independence of $T_1,\ldots,T_{k-1}$ is violated.}.
For such a $C$, Theorem \ref{thm-nucleus} gives a nucleus $K$ such that the corresponding 
nucleus identity is strongly minimal. The structure of $C$ is very strongly
represented by $C'$. As a bonus, we actually end
up greatly simplifying the polynomial-time PIT algorithm of Kayal \& Saxena~\cite{KS07}
(although we will not discuss this point in detail in this paper).
%

\subsubsection*{Proof Idea for Theorem \ref{thm-nucleus}}

The first two properties in the theorem statement are already satisfied by mat-nucleus of $C$.
So we incrementally add linear forms to the space mat-nucleus till it satisfies property (3) 
and becomes the nucleus. The addition of linear forms is guided by the ideal version of Chinese
remaindering. For convenience assume $T_1,T_2,T_3$ to be linearly independent. Then, by 
homogeneity and equal degree, we have an equivalent ideal statement: 
$T_2\notin\ideal{T_1}$ and $T_3\notin\ideal{T_1,T_2}$ (see Lemma \ref{lem-homo-ideal}). Even in
this general setting the path analogy (used in the last subsection) works and we essentially 
get linear forms $v_1\in L(T_1)$ and $v_2\in L(T_2)$ such that: $T_2\notin\ideal{v_1}$ and $T_3\notin\ideal{v_1,v_2}$. 
We now add these forms $v_1, v_2$ to the space mat-nucleus, and 
call the new space $K$. It is expected that the new $K_1, K_2, K_3$ are now linearly 
independent. 

Not surprisingly, the above argument has numerous technical problems. But it can be made 
to work by careful applications of the ideal version of Chinese
remaindering. We give the full proof in Section \ref{sec-nucleus}.

\subsection{Step 3: Invoking Sylvester-Gallai Theorems}
\label{sec-step-3}

We make a slight, but hopefully interesting, detour and leave
depth-$3$ circuits behind. We rephrase the standard Sylvester-Gallai theorems in terms of
\emph{Sylvester-Gallai closure} (or configuration) and \emph{rank bounds}.
This is far more appropriate
for our application, and seems to be very natural in itself.

\begin{definition}[$\sg_k$-closed]\label{def-sg-1} 
Let $k \in {\NN}^{>1}$.
Let $S$ be a subset of non-zero vectors in $\FF^{n}$
without \emph{multiples}: no two vectors in $S$ are scalar
multiples of each other\footnote{This is just a set of 
elements in the projective space $\FF\PP^{n-1}$, but this
formulation in terms of vectors is more convenient for our applications.}.
Suppose that for every set $V$ of $k$ linearly independent vectors in $S$,
the linear span of $V$ contains at least $(k+1)$ vectors of $S$.
Then, the set $S$ is said to be \emph{$\sg_k$-closed}.
\end{definition}

We would expect that if $S$ is finite then it will get harder to keep $S$ $\sg_k$-closed as
$\rk(S)$ is gradually increased. This intuition holds up when $\FF=\RR$. 
As we mentioned earlier, the famous Sylvester-Gallai Theorem states: if a finite $S\subset\RR^n$ is 
$\text{SG}_2$-closed, then $\rk(S)\le2$. It is optimal as the line 
$S:=\{(1,0), (1,1), (1,2)\}$ has rank $2$ and is $\text{SG}_2$-closed. 

In fact, there is
also a generalization of the Sylvester-Gallai theorem known (as stated in Theorem 2.1 of
\cite{BE67}) : {\em Let $S$ be a finite set in $\RR\PP^{2t}$ spanning that projective space.
Then, there exists a $t$-flat $H$ such that 
$|H\cap S|=t+1$, and $H$ is spanned by those points $H\cap S$.} 
 
\comment{
(as stated in Theorem 6.2 of
\cite{KSar09}) : {\em Let $S$ be a finite set of points  
spanning an affine space $V\subseteq\RR^n$ such that $\rk(V)>2t$.
Then, there exists a $t$ dimensional hyperplane $H$ such that
$|H\cap S|=t+1$, and $H$ is spanned by the points
of $S$, i.e $\text{affine-span}(H\cap S)=H$. } 
}

Let $S\subset\RR^n$ be a finite set of points with first coordinate being $1$ and let $k\geq2$. 
We claim that if $S$ is $\sg_k$-closed, then $\rk(S)\le2(k-1)$. Otherwise
the above theorem guarantees $k$ vectors $V$ in $S$ whose $(k-1)$-flat $H$ has only 
$k$ points of $S$. If $\lrsp(V)$ has a point $s\in S\setminus V$ then as $S$ has first 
coordinates $1$, it would mean that a {\em convex} linear combination of $V$ (i.e. sum of 
coefficients in the combination is $1$) is $s$. In other words, $s\in H$, which contradicts
$H$ having only $k$ points of $S$. Thus, $\lrsp(V)$ also 
has no point in $S\setminus V$, but this contradicts $\sg_k$-closure of $S$. This shows that
higher dimensional Sylvester-Gallai theorem implies that if $S$ is $\sg_k$-closed then 
$\rk(S)\le2(k-1)$. We prefer using this
rephrasal of the higher dimensional Sylvester-Gallai Theorem.
This motivates the following definitions.

\begin{definition}[SG operator]\label{def-sg-2}
Let $k, m \in {\NN}^{>1}$.

{\bf [$\sg_k(\cdot,\cdot)$]}
The largest possible rank of an $\sg_k$-closed set of at most $m$ points in $\FF^n$ is denoted by 
$\sg_k(\FF,m)$. For example, the above discussion entails $\sg_k(\RR,m)\le2(k-1)$ which
is, interestingly, independent of $m$. 
(Also verify that $\sg_{k}(\FF,m)\leq\sg_{k'}(\FF',m')$ for $k\le k'$, 
$m\le m'$ and $\FF\subseteq\FF'$.)

{\bf [$\sg_k(\cdot)$]}
Suppose a set $S\subseteq\FF^n$ has rank greater than $\sg_k(\FF,m)$ (where $\#S\le m$). 
Then, by definition,
$S$ is not $\sg_k$-closed. In this situation we say the 
\emph{$k$-dimensional Sylvester-Gallai operator} $\sg_k(S)$ (applied on $S$) returns
a set of $k$ linearly independent vectors $V$ in $S$ whose span has no point in $S\setminus V$. 	
\end{definition}  

The Sylvester-Gallai theorem in higher dimensions can now be expressed
succintly.

\begin{theorem}[High dimension Sylvester-Gallai for $\RR$]\cite{Han65,BE67}\label{thm-sg}
$\sg_k(\RR,m)\le2(k-1)$. 
\end{theorem} 
\noindent {\bf Remark.} This theorem is also optimal, for if we set $S$ to be a union of $(k-1)$
``skew lines'' then $S$ has rank $2(k-1)$ and is $\sg_k$-closed. For example, when 
$k=3$ define $S:=\{(1,1,0,0),(1,1,1,0),(1,1,2,0)\}$ $\cup$ 
$\{(1,0,1,0),(1,0,1,1),(1,0,1,2)\}$. It is
easy to verify that $\rk(S)=4$ and the span of every three linearly independent vectors in
$S$ contains a fourth vector!

Using some linear algebra and combinatorial tricks, we prove the 
first ever Sylvester-Gallai bound for all fields. The proof is in Section~\ref{sec-sg},
where there is a more detailed discussion of this (and the connection
with LDCs).

\smallskip\noindent
{\bf Theorem \ref{thm-sgk}} ($\sg_k$ for all fields). For any field $\FF$ and 
$k,m\in\NN^{>1}$, $\sg_k(\FF,m)\leq 9k\lg m$.

\subsubsection{Back to identities}

Let $C$ be a simple and strongly minimal $\sps(k,d)$ identity.
Theorem \ref{thm-nucleus} gives us a nucleus $K$, of rank $<2k^2$, that matches $T_1$ to each 
term $T_i$. As seen in Step 2, if we look at the corresponding
multiplication terms $K_i:=M(L_K(T_i))$, $i\in[k]$, then they again form a $\sps(k,d')$
``nucleus identity'' $C'=\sum_{i\in[k]}\alpha_iK_i$, for some $\alpha_i$'s in $\FF^*$, which is 
simple and strongly minimal. Define the {\em non-nucleus part} of $T_i$ as 
$L_K^c(T_i):=L(T_i)\setminus K$, for all $i\in[k]$ ($c$ in the exponent annotates
``complement", since $L(T_i)=L_K(T_i)\sqcup L_K^c(T_i)$). What can we say about the rank of 
$L_K^c(T_i)$ ? 

Define the {\em non-nucleus part of $C$} as $L^c_K(C):=$ $\bigcup_{i\in[k]}L_K^c(T_i)$.
Our goal in Step 3 is to bound $\rk(L^c_K(C)\ \mod K)$ by $2k$ 
when the field is $\RR$. This will give us a rank bound of $\rk(K)+$ $\rk(L^c_K(C) \mod K)$ 
$<(2k^2+2k)$ for simple and strongly 
minimal $\sps(k,d)$ identities over $\RR$. The proof is mainly combinatorial, based on higher
dimensional Sylvester-Gallai theorems and a property of set partitions, with a sprinkling
of algebra. 

We will finally apply $\sg_k$ operator not directly on the forms in $L(C)$ but on a suitable
truncation of those forms. So we need another definition.

\begin{definition}[Non-$K$ rank]\label{def-non-nucleus} 
Let $K$ be a linear subspace of $L(R)$. Then $L(R)/K$ is again a linear space (the {\em quotient 
space}). Let $S$ be a list of forms in $L(R)$. The {\em non-$K$ rank of $S$} is defined to
be $\rk(S \mod K)$ (i.e. the rank of $S$ when viewed as a subset of $L(R)/K$).  

Let $C$ be a $\sps(k,d)$ identity with nucleus $K$. The non-$K$ rank of the non-nucleus part 
$L^c_K(T_i)$ is called the {\em non-nucleus rank of $T_i$}. 
Similarly, the non-$K$ rank of the non-nucleus part $L^c_K(C):=$
$\bigcup_{i\in[k]}L^c_K(T_i)$ is called the {\em non-nucleus rank of $C$}.
\end{definition}

We give an example to explain the non-$K$ rank. 
Let $R = \FF[z_1,\cdots,z_n,y_1,\cdots,y_m]$. Suppose
$K = \lrsp(z_1,\cdots,z_n)$ and $S \subset L(R)$. We can
take any element $\ell$ in $S$ and simply drop all the $z_i$ terms, i.e. `truncate' $z$-part of 
$\ell$. This
gives a set of linear forms over the $y$ variables. The rank
of these is the non-$K$ rank of $S$.

We are now ready to state the theorem that is proved in Step 3. It basically shows a 
neat relationship between the non-nucleus part and Sylvester-Gallai.

\begin{theorem}[Bound for simple, strongly minimal identities]\label{thm-strong-rank} 
Let $|\FF|>d$. The non-nucleus rank of 
a simple and strongly minimal $\sps(k,d)$ identity over $\FF$ is at most $\sg_{k-1}(\FF,d)$.
\end{theorem}	

Given a simple, minimal $\sps(k,d)$ identity $C$ that is not strongly minimal. 
Let $T_1,\ldots,T_{k'}$ be linearly independent and form a basis of $\{T_i|i\in[k]\}$. Then 
it is clear that $\exists\vec{a}\in\FF^{k'}\setminus\{\vec{0}\}$ such that 
$\sum_{i\in[k']}a_iT_i+T_{k'+1}$ is a strongly minimal $\sps(k'',d)$ identity (for some 
$1<k''\leq k'+1$). 
Hence, we could apply the above theorem on this identity and get a rank bound for the non-nucleus 
part. The only problem is this fanin-$k''$ identity may not be simple. Our solution
for this is to replace $T_{k'+1}$ by the suitable linear combination of $\{T_i|i\in[k']\}$
in $C$ and repeat the above argument on the new identity. In Section \ref{sec-gen-case}
we show this takes care of the
whole non-nucleus part and bounds its rank by $k\cdot\sg_k(\FF,d)$. To state the theorem formally,
we need a more refined notion than the fanin of a $\sps$ circuit.

\begin{definition}[Independent-fanin]
Let $C=\sum_{i\in[k]}T_i$ be a $\sps(k,d)$ circuit. The {\em independent-fanin} of $C$, $\fanrk(C)$, is defined
to be the size of the maximal $\cI\subseteq[k]$ such that $\{T_i|i\in\cI\}$ are linearly independent
polynomials. (Remark: If $\fanrk(C)=k$ then $C\ne0$. Also, for an identity $C$, $C$ is strongly minimal
iff $\fanrk(C)=k-1$.)
\end{definition}

We now state the following stronger version of the main theorem.

\begin{theorem}[Final bound]\label{thm-rank} Let $|\FF|>d$. The rank of a simple, minimal
$\sps(k,d)$, independent-fanin $k'$, identity is at most $2k^2 + (k-k')\cdot\sg_{k'}(\FF,d)$.
\end{theorem}
\noindent {\bf Remark:}
In particular, the rank of a simple, minimal $\sps(k,d)$ identity over reals is at most 
$2k^2 + (k-k')\cdot\sg_{k'}(\RR,d)$ $\leq$ $2k^2 + (k-k')2(k'-1)$ $<3k^2$, proving the main theorem
over reals. Likewise, for any $\FF$, we get the rank bound of 
$2k^2 + (k-k')\cdot\sg_{k'}(\FF,d)$ $\leq$ $2k^2 + (k-k')9k'\lg d\le$ $2k^2 + \frac{9k^2}{4}\lg d$
$<3k^2\lg 2d$, proving the main theorem.

\subsubsection*{Proof Idea for Theorem \ref{thm-strong-rank}}

Basically, we apply the $\sg_k(\cdot)$ operator on the non-nucleus part of the term $T_1$, 
i.e. we treat a linear form $\sum_ia_ix_i$ as the point 
$(1,\frac{a_2}{a_1},\ldots,\frac{a_n}{a_1})\in\FF^n$ for the purposes of Sylvester-Gallai and then 
we consider $\sg_k(L_K^c(T_1))$ assuming that the non-nucleus rank of $T_1$ is more
than $\sg_k(\FF,d)$. 
This application of Sylvester-Gallai is much more direct compared
to the methods used in \cite{KSar09}. There, they needed versions
of Sylvester-Gallai that dealt with colored points and had to
prove a {\em hyperplane decomposition} property after applying essentially a
$\sg_{k^{O(k)}}(\cdot)$ operator on $L(C)$. 
Since, modulo the nucleus, all multiplication terms look essentially the same,
it suffices to focus attention on just one of them.
Hence, we apply the $\sg_k$-operator on a single multiplication term.

To continue with the proof idea, assume $C$ is a simple, strongly minimal $\sps(k,d)$
identity with terms $\{T_i|i\in[k]\}$ and let $K$ be its nucleus given by Step 2. It will
be convenient for us to fix a linear form $y_0\in L(R)^*$ and a subspace $U$ of $L(R)$ 
such that we have the following {\em orthogonal} vector space decomposition 
$L(R)=\FF y_0\oplus U\oplus K$ (i.e. $\ell\in \FF y_0\cap U$ implies $\ell=0$ and
$\ell\in (\FF y_0\oplus U)\cap K$ implies $\ell=0$). This means for any form $\ell\in L(R)$, 
there is a unique way to express $\ell=\alpha y_0+u+v$, where $\alpha\in\FF$, $u\in U$ and 
$v\in K$.
Furthermore, we will assume wlog that for every form $\ell\in L^c_K(T_1)$ the corresponding 
$\alpha$ is nonzero, i.e. each form in $L^c_K(T_1)$ is {\em monic} wrt $y_0$ (see Lemma
\ref{lem-monic-forms}).

\begin{definition}[$\trun(\cdot)$]\label{def-trun} 
Fix a decomposition $L(R)=\FF y_0\oplus U\oplus K$. For any form $\ell\in L^c_K(T_1)$, 
there is a unique way to express $\ell=\alpha y_0+u+v$, where $\alpha\in\FF^*$, $u\in U$ and 
$v\in K$.
 
The {\em truncated form $\trun(\ell)$} is the linear form obtained by dropping the $K$ part and 
normalizing, i.e. $\trun(\ell):=y_0+\alpha^{-1}u$. 

Given a list of forms $S$ we define $\trun(S)$ to be the corresponding {\em set} (thus no 
repetitions) of truncated forms. 
\end{definition}

To be precise, we fix a basis $\{y_1,\ldots,y_{\rk(U)}\}$ of $U$ so that each form in $\trun(L_K^c(T_1))$
has representation $y_0+\sum_{i\geq1}a_iy_i$ ($a_i\text{'s}\in\FF$). We view each such form as the 
{\em point}
$(1,a_1,\ldots,a_{\rk(U)})$ while applying Sylvester-Gallai on $\trun(L_K^c(T_1))$.
Assume, for the sake of 
contradiction, that the non-nucleus rank of $T_1$, $\rk(\trun(L_K^c(T_1)))>\sg_k(\FF,d)$ then 
(by definition) 
$\sg_k(\trun(L_K^c(T_1)))$ gives $k$ linearly independent forms $\ell_1,\ldots,\ell_k\in(y_0+U)$
whose span contains no {\em other} linear form of $\trun(L_K^c(T_1))$.

For simplicity of exposition, let us fix $k=4$, $K$ spanned by $z$'s, $U$ spanned by $y$'s
and $\ell_i=y_0+y_i$ $(i\in[4])$. Note that (by definition)
$\trun(\alpha y_0+\sum_i \alpha_i z_i + \sum_i \beta_i y_i) = y_0+\sum_i\frac{\beta_i}{\alpha} y_i$. 
We want to derive a
contradiction by using the $\sg_4$-tuple $(y_0+y_1,y_0+y_2,y_0+y_3,y_0+y_4)$ {\em and} the fact 
that $C$ is a simple, strongly minimal $\sps(4,d)$ identity. The contradiction is easy to see
in the following configuration: Suppose the linear forms in $C$ that are similar to a form in 
$\bigcup_{i\in[4]}(y_0+y_i+K)$ are {\em exactly} those depicted in the figure. Let us consider
$C$ modulo the ideal $I:=\ideal{y_0+y_1+z_1, y_0+y_2+z_2, -y_0-y_4+z_2}$. It is easy to see
that these forms (call them $\ell_1',\ell_2',\ell_4'$) ``kill" the first three gates, leaving 
$C\equiv T_4 (\mod I)$. As $C$ is an identity this means $T_4\in I$, thus there is a form 
$\ell\in L(T_4)$ such that $\ell\in\lrsp(\ell_1',\ell_2',\ell_4')$. Now none of the forms 
$\ell_1',\ell_2',\ell_4'$ divide $T_4$. 
Also, their non-trivial combination, say 
$\alpha\ell_1'+\beta\ell_2'$ for $\alpha\beta\ne0$, cannot occur in $L(T_4)$.
Otherwise, by the matching property $\trun(\alpha\ell_1'+\beta\ell_2')=$ 
$(\alpha+\beta)^{-1}(\alpha\ell_1+\beta\ell_2)$ will be 
in $\trun(L^c_K(T_1))$. This contradicts the $\ell_i$'s being a
$\sg_4$-tuple. Thus, $T_4$ cannot be in $I$, a contradiction. This means
that the non-nucleus rank of $T_1$ is $\leq\sg_4(\FF,d)$, which by matching properties implies
the non-nucleus rank of $C$ is $\leq\sg_4(\FF,d)$.

\begin{center}\includegraphics[scale=0.48]{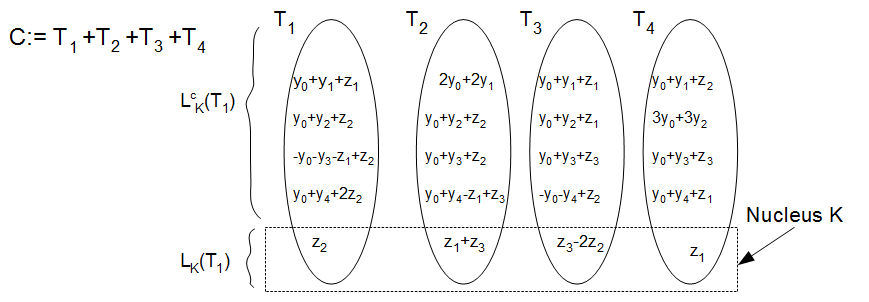}\end{center}

We were able to force a contradiction because we used a set of forms in
an SG-tuple that killed three terms and ``preserved" the last term.
Can we always do this? This is not at all obvious, and that is because
of repeating forms. Suppose, after going modulo form $\ell$,
the circuit looks like $x^3y + 2x^2y^2 + xy^3 = 0$. This is not simple,
but \emph{it does not have to be}. We are only guaranteed that the original
circuit is simple. Once we go modulo $\ell$, that property is lost.
Now, the choice of \emph{any} form kills all terms. In the figure
above, $\<y_0+y_1+z_1, y_0+y_2+z_2$, $y_0+y_3+z_3\>$ does not yield
a contradiction. We will use our more powerful Chinese remaindering
tools and the nucleus properties to deal with this.
We have to prove 
a special theorem about partitions of $[k]$ and use strong minimality (which we did not use in the 
above sketch). The full proof is given in Section \ref{sec-strong}. 


\section{Matching the Terms in an Identity: Construction of mat-nucleus}
\label{sec-mat-nucleus}

\subsection{Chinese Remaindering for Multiplication Terms}

Traditionally, Chinese remaindering is the fact: if two coprime polynomials (resp. integers) $f, g$ 
divide a polynomial (resp. integer) $h$ then $fg$ divides $h$. The key tool in constructing mat-nucleus 
is a version of Chinese 
remaindering specialized for multiplication terms but generalized to ideals. Similar methods appeared 
first in \cite{KS07} but we turn those on their head and give a ``simpler'' proof. In particular, we avoid the 
use of local rings and Hensel lifting.

\begin{definition}[Radical-span]
Let $S:=\{f_1,\ldots,f_m\}$ be multiplication terms generating an ideal $I$. 
Define linear space $\radsp(S):=$ $sp(L(f_1)\cup\ldots\cup L(f_m))$. 

When the set of generators $S$ are clear from the context we will also use the notation $\radsp(I)$.
Similarly, $\radsp(I,f)$ would be a shorthand for $\radsp(S\cup\{f\})$.
\end{definition}
\noindent {\bf Remark.} Radical-span is motivated by the {\em radical} of an ideal but it is not quite that, for example, 
$\text{radical}(x_1^2,x_1x_2)=$ $\ideal{x_1}$ but $\radsp(x_1^2,x_1x_2)=sp(x_1,x_2)$. It is easy to see that 
the ideal generated by $\radsp$ always contains the radical ideal. 

Now we can neatly state Chinese remaindering as an {\em ideal decomposition} statement.

\begin{theorem}[Ideal Chinese remaindering] \label{thm-crt}
Let $f_1,\ldots,f_m, z,f,g$ be multiplication terms. Define the ideal $I:=\ideal{f_1,\ldots,f_m}$. 
Assume $L(z)\subseteq \radsp(I)$ while, $L(f)\cap\radsp(I)=\emptyset$ and $L(g)\cap\radsp(I,f)=\emptyset$. Then, 
$\ideal{I, zfg}=\ideal{I, z} \cap \ideal{I, f} \cap \ideal{I, g}$.
\end{theorem}
\begin{proof}
If $h$ is a polynomial in $\ideal{I, zfg}$ then clearly it is in each of the ideals $\ideal{I, z}$, $\ideal{I, f}$ and 
$\ideal{I,g}$.

Suppose $h$ is a polynomial in $\ideal{I, z} \cap \ideal{I, f} \cap \ideal{I, g}$. Then by definition there exist 
$i_1,i_2,i_3\in I$ and $a,b,c\in R$ such that,
$$h=i_1+az=i_2+bf=i_3+cg.$$

The second equation gives $bf\in \ideal{I,z}$. Since $L(f)\cap\radsp(I,z)=L(f)\cap\radsp(I)=\emptyset$, 
repeated applications of Lemma \ref{lem-non-zd} give us, $b\in\ideal{I,z}$.
Implying $bf \in \ideal{I,z} f \subseteq \ideal{I, zf}$, hence $h=i_2+bf\in \ideal{I,zf}$.
This ensures the existence of $i_2'\in I$ and a polynomial $b'$ such that,
$$h=i_2'+b'zf=i_3+cg.$$

Again this system says that $cg\in \ideal{I,zf}$. Since $L(g)\cap\radsp(I,zf)=L(g)\cap\radsp(I,f)=\emptyset$, 
repeated applications of Lemma \ref{lem-non-zd} give us $c\in\ideal{I,zf}$.
Implying $cg \in \ideal{I,zf} g\subseteq \ideal{I, zfg}$, 
hence $h=i_3+cg\in \ideal{I,zfg}$. This finishes the proof.
\end{proof}

The conditions in this theorem suggest that factoring a multiplication term $f$ into parts corresponding
to the equivalence classes of ``similarity mod $\radsp(I)$" would be useful. 

\begin{definition}[Nodes]
Let $f$ be a multiplication term and let $I$ be an ideal generated by some multiplication terms.
As the relation ``similarity mod $\radsp(I)$" is an equivalence relation on $L(R)$, it partitions, in 
particular, the list $L(f)$ into equivalence classes. 

{\bf[$\piv_I(f)$]}
For each such class pick a representative $\ell_i$
and define $\piv_I(f):=$ $\{\ell_1,\ldots,\ell_r\}$.
(Note that form $0$ can
also appear in this set, it represents the class $L(f)\cap\radsp(I)$.)

{\bf[$\nod_I(f)$]}
For each $\ell_i \in \piv_I(f)$, we multiply the forms in $f$ that are similar to 
$\ell_i$ 
mod $\radsp(I)$. We define {\em nodes of $f$ mod $I$} as the set of polynomials $\nod_I(f):=$ 
$\{M(L(f)\cap(\FF^*\ell+\radsp(I)))\ |\ \ell\in\piv_I(f)\}$. (Remark: When $I=\{0\}$, nodes of $f$ 
are just the coprime powers-of-forms dividing $f$.)

{\bf [...wrt a subspace]}
Let $K$ be a linear subspace of $L(R)$. Clearly, the relation 
``similarity mod $K$" is an equivalence relation on $L(R)$. It will be convenient for us to also 
use notations $\piv_K(f)$ and $\nod_K(f)$. They are defined by replacing $\radsp(I)$ in the
above definitions by $K$.
\end{definition}

Observe that the product of polynomials in $\nod_I(f)$ just gives $f$.
Also, modulo $\radsp(I)$, each node is just a form-power $\ell^r$. In other words,
modulo $\radsp(I)$, a node is rank-one term. The choice of the word ``node"
might seem a bit mysterious, but we will eventually construct paths through these.
To pictorially see what is going on, think of each term $T_i$ as a set of 
its constituent nodes.

We prove a corollary of the ideal Chinese remaindering
theorem that will be very helpful in both Steps 1 and 2.

\begin{corollary}\label{cor-crt}
Let $h\in R$, $f$ be a multiplication term, and let $I$ be an ideal generated by some multiplication 
terms. Then, $h\notin\ideal{I,f}$ iff $\exists g\in\nod_I(f)$ such that $h\notin\ideal{I,g}$.
\end{corollary}
\begin{proof}
If $h\notin\ideal{I,g}$, for some $g\in\nod_I(f)$, then clearly $h\notin\ideal{I,f}$.

Conversely, assume $h\notin\ideal{I,f}$.
Let $\piv_I(f)=$ $\{\ell_1,\ldots,\ell_r\}$ and correspondingly, $\nod_I(f)=$ $\{g_1,\ldots,g_r\}$.
If $r=1$ then $f$ is similar to $g_1$, hence $h\notin\ideal{I,g_1}$ and we are done. 
So assume
$r\geq2$. Also, in case $L(f)$ has a form in $\radsp(I)$, assume wlog $\ell_1$ is the representative of
the class $L(f)\cap\radsp(I)$. Define $G_i:=\prod_{i<j\le r}g_i$, for all $i\in[r-1]$.

We claim that for all $i\in[r-1]$, $L(G_i)\cap\radsp(I,g_i)=\emptyset$. Otherwise 
$\exists\ell\in L(G_i)$ such that either $\ell\in\radsp(I)$ or $\ell\in(\FF^*\ell_i+\radsp(I))$.
Former case contradicts $\ell_1$ being the representative of the class $L(f)\cap\radsp(I)$, while
the latter case contradicts $\ell_{i+1},\ldots,\ell_r$ being {\em non}-similar to $\ell_i$
mod $\radsp(I)$. Thus, for all $i\in[r-1]$, $L(G_i)\cap\radsp(I,g_i)=\emptyset$, and by 
applying Theorem \ref{thm-crt} on $\ideal{I,g_iG_i}$ for each $i\in[r-1]$, we deduce:
$$\ideal{I,f}=\left\<I,\prod_{i\in[r]}g_i\right\>=\bigcap_{i\in[r]}\ideal{I,g_i}.$$
Thus, $h\notin\ideal{I,f}$ implies the existence of some $i\in[r]$ such that $h\notin\ideal{I,g_i}$. 
\end{proof}

\subsection{Applying Chinese Remaindering to $\sps$ Circuits}

We showed the effect of ideal Chinese remaindering on a single multiplication
term $f$ in Corollary \ref{cor-crt}. Now we show the effect on a {\em tuple} 
of multiplication terms, for example, appearing in a $\sps$ circuit. We then
need, quite naturally, a notion of {\em path} of nodes.
 
\begin{definition}[Paths]
Let $I$ be an ideal generated by some multiplication terms.
Let $C=\sum_{i\in[k]}T_i$ be a $\sps(k,d)$ circuit. Let $v_i$ be a {\em sub-term} of 
$T_i$ (i.e. $L(v_i)\subseteq L(T_i)$), for all $i\in[k]$. 
We call the tuple $(I,v_1,\ldots,v_k)$ a {\em path of $C$ mod $I$} if, for all $i\in[k]$,
$v_i\in\nod_{\ideal{I,v_1,\ldots,v_{i-1}}}(T_i)$. It is of {\em length} $k$.
(Remark: We have defined path $\vec{p}$ as a tuple but, for convenience, we will sometimes
treat it as a {\em set} of multiplication terms, eg. when operated upon by $\lrsp(\cdot)$, 
$\ideal{\cdot}$, $\radsp(\cdot)$, etc.)

Conventionally, when $k=0$ the circuit $C$ has just ``one'' gate: $0$. In that case,
the only path $C$ has is $(I)$, which is of length $0$.

We also define, for any subset $S\subseteq[k]$, the {\em sub-circuit} $C_S:=\sum_{s\in S}T_s$. 

For an $i\in\{0,\ldots,k-1\}$, define $[i]':=[k]\setminus[i]$. 
We set $[0]:=\emptyset$ and $C_\emptyset:=0$. 
\end{definition}

We now show that if $C$ is a nonzero $\sps(k,d)$ circuit then $\exists i\in\{0,\ldots,k-1\}$,
such that $C_{[i]}$ has a path $\vec{p}$ for which, $C\ (\mod \ideal{\vec{p}})$ is nonzero {\em and} 
similar to some multiplication term. This rather special path inside $C$ can be seen as a
certificate for the nonzeroness. 
The rank of the linear forms
appearing in this path can be at most $i+\rk(\radsp(I))$, since
the rank of each node is one, modulo the radical-span of the previous nodes in the path.
Hence, it is a {\em low-rank certificate for the 
nonzeroness of $C$}.  

\begin{theorem}[Certificate for a Non-identity]\label{thm-cert-non0}
Let $I$ be an ideal generated by some multiplication terms. Let $C=\sum_{i\in[k]}T_i$ be a 
$\sps(k,d)$ circuit that is nonzero modulo $I$. Then $\exists i\in\{0,\ldots,k-1\}$ such that 
$C_{[i]}$ mod $I$ has a path $\vec{p}$ satisfying: 
$C_{[i]'}\equiv \alpha\cdot T_{i+1}\not\equiv0$ $(\mod \vec{p})$ for some 
$\alpha\in\FF^*$.
\end{theorem}

Before we prove the theorem, we make an aside observation.
If the reader has kept the mental picture of the
terms as consisting of rank-one (modulo $\radsp(I)$) nodes, then the notion of a path
has some meaning. A path $\vec{p}$ kills the terms that is passed through, and
collapses remaining circuit into a single term. 
This is very reminiscent of the poly-time algorithm of Kayal \& Saxena~\cite{KS07}.
Indeed, this theorem is a (shorter) \emph{proof} of the correctness
of the algorithm. Why? Consider the path $\vec{p}$ given by the theorem
when $I$ is the zero ideal.
The path $\vec{p}$ can be represented by a list of at most $k$ `forms' in $L(C)$.
This path comes from some $C_{[i]}$, which means that $C_{[i]} = 0 (\mod \vec{p})$.
So, we get that $C \equiv \alpha \cdot T_{i+1} \not\equiv 0 (\mod \vec{p})$.
Since $T_{i+1}$ is a product of linear
forms, it is easy to algorithmically check if $C \equiv 0 (\mod \vec{p})$.
If $C$ is identically zero, such a path cannot exist. Since there are at most
$d^k$ different paths, we can exhaustively check all of them. That yields
an alternative view of \cite{KS07}'s test.

\begin{proof}
Fix an $i\in\{0,\ldots,k-1\}$ and a path $\vec{p}$ of $C_{[i]}$ mod $I$ such that: 

1) $C_{[i]'}\notin\ideal{\vec{p}}$ and, 

2) the set $J_i:=\{j\in[i]'\ |\ T_j\notin\ideal{\vec{p}}\}\ne\emptyset$ is the smallest possible
(over all $i$). 

\smallskip\noindent  
Note that for values $i=0$, $\vec{p}=(I)$, the condition (1) is satisfied and the corresponding
$J_i\ne\emptyset$. Thus, there also exist $i$ and $\vec{p}$ satisfying both the conditions (1) and (2).

Let $j^*$ be the smallest element in $J_i$. This means that for every $m$, $i<m<j^*$, 
$T_m\in\ideal{\vec{p}}$. This means, by repeated applications of Lemma \ref{lem-non-zd}, 
$v_m:=M(L_{\radsp(\vec{p})}(T_m))$ $\in\ideal{\vec{p}}$. Thus, $\ideal{\vec{p}}=$
$\ideal{\vec{p}\cup\{v_m|i<m<j^*\}}$. This makes $\vec{q}:=(\vec{p},(v_m|i<m<j^*))$ also a path
of $C_{[j^*-1]}$ mod $I$. We now claim that $\vec{q}$ is the path promised in the theorem statement.

Note that $C_{[j^*-1]'}\equiv C_{[i]'}$ $(\mod \vec{p})$ and $C_{[i]'}\notin\ideal{\vec{p}}=$ 
$\ideal{\vec{q}}$, in other words,
path $\vec{q}$ also satisfies: 

1) $C_{[j^*-1]'}\notin\ideal{\vec{q}}$ and, 

2) the set $J_{j^*-1}$ $=\{j\in[j^*-1]'\ |\ T_j\notin\ideal{\vec{q}}\}$ $=J_i$ is still 
the smallest possible. 

\smallskip\noindent  
If $C_{[j^*-1]'}$ 
$\notin\ideal{\vec{q},T_{j^*}}$ then, by Corollary \ref{cor-crt}, there exists 
$v_{j^*}\in\nod_{\ideal{\vec{q}}}(T_{j^*})$ such that $C_{[j^*-1]'}$ $\notin\ideal{\vec{q},v_{j^*}}$,
hence $C_{[j^*]'}=$ $C_{[j^*-1]'}-T_{j^*}$ $\notin\ideal{\vec{q},v_{j^*}}$.   
Define $\vec{q}':=(\vec{q},v_{j^*})$, clearly it is a path of $C_{[j^*]}$ mod $I$. Wrt this path $\vec{q}'$, 
$J_{j^*}\subseteq J_i\setminus\{j^*\}$ $\subsetneq J_i$ together with $C_{[j^*]'}$ 
$\notin\ideal{\vec{q}'}$, contradicting the minimality assumption on $i$.
Thus, we assume $C_{[j^*-1]'}$ $\in\ideal{\vec{q},T_{j^*}}$. By Lemma \ref{lem-homo-ideal}, this
guarantees the existence of an $\alpha\in\FF$ such that,
$$(C_{[j^*-1]'}-\alpha T_{j^*})\in\ideal{\vec{q}}=\ideal{\vec{p}}.$$
Since $C_{[j^*-1]'}\equiv C_{[i]'}$ $\not\equiv0$ $(\mod \vec{p})$,
the above equation can be rewritten as:
$$C_{[j^*-1]'}\equiv\alpha T_{j^*}\not\equiv0\ (\mod \ideal{\vec{q}}).$$
Thus, finishing the proof ($\alpha$ nonzero is implied).
\end{proof}

\noindent {\bf Remark.} The above theorem is quite powerful, for instance, it only needs the {\em non}-zeroness of 
$C$ mod $I$ without referring to any simplicity or minimality requirements. 

\subsection{Using Minimality to get mat-nucleus}

If we are given a circuit that is zero \& minimal (may not be simple) then a repeated application 
of Theorem \ref{thm-cert-non0} gives us a space {\em mat-nucleus} that matches all the 
multiplication terms of $C$. 

\smallskip\noindent
{\bf Theorem \ref{thm-mat-nucleus}} (Matching-Nucleus). 
Let $C=T_1+\cdots+T_k$ be a $\sps(k,d)$ circuit that is minimal and zero. 
Then there exists a linear subspace $K$ of $L(R)$ such that:
\begin{enumerate}
\item[1)] $\rk(K)< k^2$.
\item[2)] $\forall i\in[k]$, there is a $K$-matching $\pi_i$ between $T_1, T_i$. 
\end{enumerate}
\begin{proof}
The proof is an iterative process with at most $k$ rounds. We maintain a set $\cP$,
containing paths of some sub-circuits of $C$, and an undirected graph $G = ([k],E)$.
For convenience, define $U:=\radsp(p|p\in\cP)$ (i.e. consider each path $p$ as a set of 
multiplication terms, take the union of all these sets, and compute its radical-span). 
The {\em invariant} at the end of each round is: $(i,j)\in E$ iff $T_i, T_j$ are
$U$-matched. At the end of round $0$ we assume, $\cP:=\{(0)\}$ and $E:=$
$\{(i,j)\in[k]^2\ |\ T_i, T_j\text{ are similar}\}$. We want to eventually make $G$ a 
connected graph (infact a {\em $k$-clique}) by keeping $\rk(U)$ as small as possible.

Suppose the invariant holds till the end of some round $(r-1)\geq0$. If $G$ is connected
then the process stops at round $(r-1)$. Otherwise, we will show how
to {\em decrease} the number of connected components of $G$ in round $r$. 
Say, $G$ has a maximal connected component on vertices $S\subsetneq[k]$. 
Since $C_S\ne0$ (by minimality), we can apply Theorem \ref{thm-cert-non0} on $C_S$ mod $\ideal{0}$ 
to get a path $p_S$ inside $C_S$ mod $\ideal{0}$ such that $\exists i\in S$,
$C_S\equiv \alpha T_i \not\equiv0\ (\mod p_S)$ for some $\alpha\in\FF^*$. 

Define $S':=[k]\setminus S$.
Now,
\begin{equation}\label{eqn-pf-mat-nucleus} 
C \equiv C_{S'}+\alpha T_i \equiv 0\ (\mod p_S). 
\end{equation}
This means $C_{S'}\notin\ideal{p_S}$ (otherwise $\alpha T_i\in\ideal{p_S}$, a contradiction).
Thus, we can apply Theorem \ref{thm-cert-non0} on $C_{S'}$ mod $\ideal{p_S}$ to get a path 
$p_{S'}$ inside $C_{S'}$ mod $\ideal{p_S}$ such that, $C_{S'}\equiv\beta T_j\not\equiv0 (\mod p_{S'})$, 
for some $\beta\in\FF^*$. This allows us to rewrite Equation (\ref{eqn-pf-mat-nucleus}) as: 
$$\alpha T_i\equiv -\beta T_j\not\equiv0\ (\mod p_{S'})$$
Define $K':=\radsp(p_{S'})$.
As $p_{S'}$ is, after all, a path of some sub-circuit of $C\ \mod\ideal{0}$, of
length at most $|S|-1+|S'|-1=k-2$, we deduce that
$\rk(K')<(k-1)$. Also, by Lemma \ref{lem-I-match}, the above congruence implies a $K'$-matching  
between $T_i$ and $T_j$. We append the path $p_{S'}$ to $\cP$ and update $U$. Note that for 
any edge $(i,i')$ in the connected component $S$, and for any edge $(j,j')$ in the connected 
component $\tilde{S}$ (of vertex $j$): since $T_i, T_{i'}$ are still $U$-matched; 
$T_j, T_{j'}$ are still $U$-matched; $T_i,T_j$ are newly $K'$-matched; gives us that 
$T_{i'}, T_{j'}$ are newly $U$-matched. In other words, the two different connected components
$S$ and $\tilde{S}$ of $G$ will now form a bigger connected component (infact a {\em clique}) 
when we update the graph as,
$E:=\{(a,b)\in[k]^2\ |\ T_a, T_b\text{ are }U\text{-matched}\}$.

So in every round we are increasing $\rk(U)$ by at most $(k-1)$, maintaining the invariant, and
decreasing the number of connected components in $G$ by at least one. Thus, after at most
$(k-1)$ repetitions we get a $U$ that matches $T_1, T_i$, for all $i\in[k]$, and $\rk(U)<k^2$.
We define this $U$ as $K$, finishing the proof.  
\end{proof}

\section{Certificate for Linear Independence of terms: Constructing nucleus}
\label{sec-nucleus}

Suppose we have multiplication gates $T_1,\ldots,T_{k'}$ and a space $K'$ of $L(R)$
such that $T_1, T_i$ is $K'$-matched, for all $i\in[k']$. We show in this section that 
if $T_1,\ldots,T_{k'}$ are linearly independent (i.e. $\nexists$ 
$\vec{\beta}\in\FF^{k'}\setminus\{\vec{0}\}$ s.t. $\sum_{i\in[k']}\beta_iT_i=0$) then 
$K'$ can be extended to a linear space $K$ of rank at most $(\rk(K')+k'^2)$ such that: 
$M(L_K(T_1)),\ldots,M(L_K(T_{k'}))$ are also linearly independent. This will prove 
Theorem \ref{thm-nucleus}.

\smallskip\noindent
{\bf Theorem \ref{thm-nucleus}} (Nucleus).
Let $C=\sum_{i\in[k]}T_i$ be a minimal $\sps(k,d)$ identity and let $\{T_i|i\in\cI\}$ be 
a maximal set of linearly
independent terms ($1\le k':=|\cI|<k$). Then there exists a linear subspace $K$ of $L(R)$ such that:
\begin{itemize}
\item[1)] $\rk(K)< 2k^2$.
\item[2)] $\forall i\in[k]$, there is a $K$-matching $\pi_i$ between $T_1, T_i$.
\item[3)] (Define $\forall i\in\cI$, $K_i:=M(L_K(T_i))$.) The terms $\{K_i|i\in\cI\}$ 
are linearly independent.
\end{itemize}
\begin{proof}
For convenience, and wlog, assume $\cI=[k']$. 
The proof is an iterative process with at most $k'^2$ iterations, and gradually builds the 
promised space $K$. Each iteration of the process maintains a space $U$ of $L(R)$ which
is intended to grow at each step and bring us closer to $K$. For convenience, define 
$U_i:=M(L_U(T_i))$, for all $i\in[k']$. Also for each $i\in\{2,\ldots,k'\}$, define
ideal $\cI_i:=\ideal{U_1,\ldots,U_{i-1}}$.

The process has two nested iterations, or phrased 
differently,
a double induction. We will call the outer ``loop" a {\em phase}, and the inner loop a 
{\em round}.
In each round the rank of $U$ increases by at most $1$, and the $i$-th phase has at most 
$i$ rounds.
At the end of the $i$-th phase ($i\geq2$), we will ensure $T_i\notin\cI_i$.
(Remark: By Lemma \ref{lem-non-zd} this is equivalent to ensuring $U_i\notin\cI_i$,
which by Lemma \ref{lem-homo-ideal} means that $U_i$ is linearly independent of 
$U_1,\ldots,U_{i-1}$.)

In the first phase we set $U:=K'$, where $K'$ is the matching-nucleus obtained by applying
Theorem \ref{thm-mat-nucleus} on $C$. This immediately gives us property (2) promised
in the theorem statement, i.e. the matching property. Also, $\rk(U)<k^2$ at the end of
the first phase.

Now the second phase. As $T_1, T_2$ are linearly independent, we get, by Lemma
\ref{lem-homo-ideal}, that $T_2\notin\ideal{ T_1 }$. By an application of Corollary 
\ref{cor-crt}, $\exists v\in\nod_{\ideal{0}}(T_1)$ such that $T_2\notin\ideal{v}$. 
We update $U\leftarrow(U+\radsp(v))$. Note that after updation 
$T_2\notin\ideal{U_1}=\cI_2$ (otherwise $T_2\in\ideal{U_1}\subseteq\ideal{v}$, since $v|U_1$).   

Now, for the $i>2$ phase. Inductively, we assume that $\forall r<i$, $T_r\notin\cI_r$
(remember that all these ideals are wrt the {\em current} $U$).
The phase consists of various rounds. At the end of the $j$-th round ($1\le j<i$),
we just want to ensure $T_i\notin\ideal{U_1,\ldots,U_j,T_{j+1},\cdots,T_{i-1}}$.
So we do nothing in the $j$-th round unless this is violated. What do we do when it is violated?

\begin{claim}\label{clm-extend} 
Let $i>2$ and $1\le j<i$. Suppose $\forall r<i$, 
$T_r\notin\ideal{U_1,\ldots,U_{r-1}}$.
Suppose $T_i\in\ideal{U_1,\cdots,U_j,T_{j+1},\cdots,T_{i-1}}$ 
but $T_i\notin \ideal{U_1,\cdots,U_{j-1},T_j,\cdots,T_{i-1}}$. There exists a 
$v\in\nod_{\ideal{U_1,\cdots,U_{j-1}}}(T_j)$
such that for the updated $U'\leftarrow(U+\radsp(v))$ we have
$T_i\notin\ideal{U_1',\cdots,U_j',T_{j+1},\cdots,T_{i-1}}$.
\end{claim}
\claimproof{clm-extend}{
Since $T_i \in \ideal{U_1,\cdots,U_j,T_{j+1},\cdots,T_{i-1}}$, by Lemma \ref{lem-homo-ideal},
we get 
$T_i + \sum_{r=j+1}^{i-1} \alpha_r T_r \in\ideal{U_1,\cdots,U_j}$ for some $\alpha_r$-s in $\FF$. 
Suppose there are two
distinct choices for $\alpha_r$-s (we will call them $\alpha_r$ and $\alpha'_r$).
Then, 
$$\left(T_i + \sum_{r=j+1}^{i-1} \alpha_r T_r\right),\  
\left(T_i + \sum_{r=j+1}^{i-1} \alpha'_r T_r\right) \in \ideal{U_1,\cdots,U_j}.$$
Subtracting,
we get $\sum_{r=j+1}^{i-1} (\alpha - \alpha'_r) T_r \in \ideal{U_1,\cdots,U_j}$.
Let $s$ be the largest index such that $\alpha_s - \alpha'_s \neq 0$.
(By the distinctness of the sequences, such an index exists.)
We get that 
$T_{s} \in \ideal{U_1,\cdots,U_j,T_{j+1},\cdots,T_{s-1}}$ $\subseteq \ideal{U_1,\cdots,U_{s-1}}$.
Since $s \leq i-1$, this contradicts the hypothesis. Hence, the sequence $\{\alpha_r\}$ is unique.

The claim hypothesis says that $T_i\notin \ideal{U_1,\cdots,U_{j-1},T_j,\cdots,T_{i-1}}$. 
That implies 
$T_i + \sum_{r=j+1}^{i-1} \alpha_r T_r$ $\notin \ideal{U_1,\cdots,U_{j-1},T_j}$.
Thus, by Corollary \ref{cor-crt}, $\exists v\in\nod_{\ideal{U_1,\cdots,U_{j-1}}}(T_j)$
such that
$T_i + \sum_{r=j+1}^{i-1} \alpha_r T_r$ $\notin \ideal{U_1,\cdots,U_{j-1},v}$.
Let us update $U$ to $U'\leftarrow(U+\radsp(v))$. (This updates $U_r$-s to $U_r'$-s.)

We now argue that $T_i \notin \ideal{U'_1,\ldots,U'_j,T_{j+1},\cdots,T_{i-1}}$.
Suppose not. Then, by Lemma \ref{lem-homo-ideal}, for some sequence $\beta_r$, 
$T_i + \sum_{r=j+1}^{i-1} \beta_r T_r \in \ideal{U'_1,\ldots,U'_j}$ 
$\subseteq\ideal{U_1,\ldots,U_j}$ (since for all $r$, $U_r|U'_r$).
By the uniqueness of $\{\alpha_r\}$, we have $\beta_r = \alpha_r$, for all $r$.
But that implies $T_i + \sum_{r=j+1}^{i-1} \alpha_r T_r$ $\in \ideal{U'_1,\ldots,U'_j}$ 
$\subseteq \ideal{U_1,\cdots,U_{j-1},v}$.
This is a contradiction and hence completes the proof.
}

Let us look at the first round (i.e. $j=1$).
Suppose $T_i \notin \ideal{U_1,T_2,\cdots,T_{i-1}}$. Then, we move directly to the 
second round, since we have already satisfied the round invariant.
Otherwise, $T_i\in\ideal{U_1,T_2,\cdots,T_{i-1}}$. Furthermore, by linear independence 
and Lemma \ref{lem-homo-ideal}, we have $T_i \notin \ideal{T_1,\cdots,T_{i-1}}$, so  
we can invoke Claim \ref{clm-extend} to get a $v\in\nod_{\ideal{0}}(T_1)$. This allows us 
to update $U\leftarrow(U+\radsp(v))$ such that
$T_i \notin \ideal{U_1,T_2,\cdots,T_{i-1}}$. 

Now for the induction step. 
We assume that, by the end of the $(j-1)$th round, 
$T_i \notin \ideal{U_1,\cdots,U_{j-1},T_{j},\cdots,T_{i-1}}$.
For the $j$-th round, either we would have to do nothing or have to apply Claim 
\ref{clm-extend} and update $U$. In either case, $\rk(U)$ increases by at most $1$.
At the end of the round, $T_i \notin \ideal{U_1,\cdots,U_{j},T_{j+1},\cdots,T_{i-1}}$.

This continues till $j = i-1$. We finally
have $T_i \notin \ideal{U_1,\cdots,U_{i-1}}$ $=\cI_i$,
giving us the required invariant for the $i$-th phase.
This completes the proof.
\end{proof}

\section{Invoking Sylvester-Gallai Theorems: The Final Rank Bound}
\label{sec-final-rank}

In this section we will bound the non-nucleus rank of a simple, minimal $\sps(k,d)$, independent-fanin $k'$, 
identity $C$ by $(k-k')\cdot\sg_{k'}(\FF,d)$. Thus, proving Theorem \ref{thm-rank}. We 
divide the proof into two subsections. First, we bound the non-nucleus rank of a 
simple, {\em strongly} minimal $\sps(k,d)$ identity $C$ by $\sg_{k-1}(\FF,d)$, finishing
the proof of Theorem \ref{thm-strong-rank}. Second, we show how to repeatedly use this 
result on a simple, minimal but not strongly-minimal identity.

\subsection{The strongly minimal case}\label{sec-strong}

Assume that $C:=\sum_{i\in[k]}T_i$ is a simple, strongly minimal $\sps(k,d)$ identity
(recall: then $T_1,\ldots,T_{k-1}$ are linearly independent polynomials).
Let $K$ be its nucleus given by Theorem \ref{thm-nucleus}. 
There are two important properties of this nucleus that we restate (and elaborate upon) for emphasis.

The first is the \emph{matching property}.
For any $i\in[k]$, $L^c_K(T_1)$ $(= L(T_1)\setminus K)$ is $K$-matched to $L^c_K(T_i)$ 
$( = L(T_i)\setminus K)$.
In other words for any $\ell\in L^c_K(T_1)$, the degrees of 
$M(L^c_K(T_1)\cap(\FF^*\ell+K))$
and $M(L^c_K(T_i)\cap(\FF^*\ell+K))$ are equal (remark: they are polynomials in $\nod_K(T_1)$ and 
$\nod_K(T_i)$ respectively). This observation motivates the following definition. 

\begin{definition}[Family]
Let $C$ be a $\sps(k,d)$ identity and $K$ be its nucleus. Let $\ell\in L^c_K(C)$. 
The {\em family of $\ell$} is defined to be the list, $\fami(\ell):=$
$\{M(L^c_K(T_i)\cap(\FF^*\ell+K))\ |\ i\in[k]\}$. Note that $\fami(\ell)$ is a multiset of size
exactly $k$, having equal degree polynomials corresponding to each term $T_i$, we fix this
ordering on the list (i.e. $i$-th element in $\fami(\ell)$ corresponds \& divides the multiplication 
term $T_i$).
\end{definition}

Verify that any two forms in $L^c_K(C)$ that are ``similar mod $K$" have the same families.

\smallskip\noindent{\bf [Partition, Class, Split \& Preserve]}
Let us focus on a list $\fami(\ell)$. The equivalence relation of similarity (i.e. mod 
$\ideal{0}$) on $\fami(\ell)$, induces a {\em partition} of $[k]$ (i.e. if $f_i,f_j\in\fami(\ell)$ 
are similar then place
$i$ and $j$ in the same partition-class). Denote this partition induced on $[k]$, by $\pa(\ell)$. 
Observe that $\pa(\ell)$ must contain at least $2$ classes (otherwise simplicity of $C$ is violated).

Each set in this partition is called a \emph{class}, and we naturally have 
a class $\cl(f)$ associated with each member of $f \in \fami(\ell)$.

We say that $\pa(\ell)$ \emph{splits} a subset $S \subseteq [k]$
if there is some class $X\in\pa(\ell)$ such that $X \cap S \neq \emptyset, S$.
Otherwise, we say that $\pa(\ell)$ {\em preserves} $S$. Note that a singleton is always
preserved.

For classes $A_1\in\pa(\ell_1)$ and $A_2\in\pa(\ell_2)$, the \emph{complement} $\vec{A_1\cup A_2}$ is just 
the set
$[k] \setminus(A_1\cup A_2)$. We will be later interested in the properties of this complement set wrt the two 
partitions.

\smallskip\noindent
The second property of the nucleus, the \emph{linear independence}, says something technical
about the nucleus identity. By definition $K_i=M(L_K(T_i))$, for all $i\in[k]$, and by Lemma
\ref{lem-nucleus} : $\sum_{i\in[k]}\alpha_iK_i=0$ for some $\alpha_i\text{-s}\in\FF^*$. 
Furthermore,
\begin{claim} \label{clm-kmin} For $1<r<k$, let $\{s_1,\cdots,s_r\}$ be
a subset $S \subsetneq [k]$, where $s_1 < s_2 < \cdots < s_r$. Then $K_{s_r}\notin\ideal{K_{s_1},\cdots,K_{s_{r-1}}}$.
\end{claim}
\begin{proof} If $s_r < k$, then this just holds from the linear independence of 
$\{K_1,\ldots,K_{k-1}\}$ and Lemma \ref{lem-homo-ideal}. 
So, we can assume $s_r = k$ and $K_k\in\ideal{K_{s_1},\cdots,K_{s_{r-1}}}$. 
By Lemma \ref{lem-homo-ideal}, this
means $K_k = \sum_{i\in[r-1]} \beta_{s_i} K_{s_i}$ for some $\beta\text{-s}\in\FF$. 
The nucleus identity gives us $K_k = -\sum_{i\in[k-1]}\frac{\alpha_i}{\alpha_k}K_i$ 
$=\sum_{i\in[r-1]}\beta_{s_i}K_{s_i}$. 
Since $r<k$, this implies that for some $\gamma$-s in $\FF$, not all
zero, $\sum_{i\in[k-1]}\gamma_i K_i = 0$. This contradicts the linear independence of 
$\{K_1,\ldots,K_{k-1}\}$, finishing the proof. 
\end{proof}

Before applying Sylvester-Gallai-type theorems (i.e. the $\sg_{k-1}$ operator) we emphasize
that, as discussed in Section \ref{sec-step-3}, there is a distinguished linear form 
$y_0\in L(R)^*$ and a subspace $U$ of $L(R)$ such that $L(R)=\FF y_0\oplus U\oplus K$  
and every form in $L^c_K(C)$ is monic wrt $y_0$. Thus, for every $\ell\in L^c_K(C)$
there exists a unique way to express : $\ell=\alpha y_0+u+v$ ($\alpha\in\FF^*$, $u\in U$
and $v\in K$). This allows us to define the truncation operator : $\trun(\ell)=y_0+$
$\alpha^{-1}u$. 

\begin{lemma}[Partitions from $\sg_{k-1}$-tuple]\label{lem-split} 
Suppose $\rk(\trun(L^c_K(T_1)))$ $>\sg_{k-1}(\FF,d)$, and 
$\sg_{k-1}(\trun(L^c_K(T_1)))$ gives the set $\{\ell_1, \ell_2,\cdots,\ell_{k-1}\}$.
For all $i\in[k-1]$, let $\ell_i'\in L^c_K(T_1)$ be a form satisfying 
$\trun(\ell_i')=\ell_i$. 

Let $\cI\subseteq[k-1]$ be nonempty, and $A_i$ be any class in $\pa(\ell_i'	)$ for all $i\in\cI$. Suppose 
$S:=\overline{\bigcup_{i\in\cI}A_i}$ $\ne\emptyset$. Then $S$ is split
by $\pa(\ell_c')$, for {\em some} $c\in\cI$. 
\end{lemma}
\begin{proof} 
We prove by contradiction.
Suppose $S$ is preserved by $\pa(\ell_i')$, for all $i\in\cI$. 
Since for all $i\in\cI$, $A_i\in\pa(\ell_i')$, by definition there exists an
$f_i\in\fami(\ell_i')$ such that $A_i=\cl(f_i)$. Similarly, for all $i\in\cI$,
there exists a $g_i\in\fami(\ell_i')$ such that $S\subseteq\cl(g_i)$. Note that, 
by definition, sets $A_i$ and $S$ are disjoint, hence the classes $\cl(f_i)$ and
$\cl(g_i)$ are different, implying $f_i, g_i$ are {\em not} similar, for all $i\in\cI$. 

Define ideal $I:=\ideal{f_i | i\in\cI}$. 
Let us focus on the sub-circuit $C_S = \sum_{j \in S} T_j$.  
Since $C=0$ and $S=\overline{\bigcup_{i\in\cI}\cl(f_i)}$, we deduce $C_S\in I$ (as
$f_i$ ``kills" the term $T_r$ for all $r\in\cl(f_i)$, and ``spares" the other terms).
For all $i\in\cI$, $S\subseteq\cl(g_i)$ we deduce that :  
$\prod_{i\in\cI}g_i$ divides $T_j$, for all $j\in S$. So $T_j':=T_j/(\prod_{i\in\cI}$ $g_i)$ 
is again a multiplication term
with none of its form in $\bigcup_{i\in\cI}(\FF^*\ell_i'+K)=$ $\bigcup_{i\in\cI}(\FF^*\ell_i+K)$. 
Thus, we get an important equation:
$$C_S\ =\ \(\prod_{i\in\cI}g_i\)\cdot\(\sum_{j\in S}T_j'\)\ \in\ \ideal{f_i\ |\ i\in\cI}. $$
By a repeated application of Lemma \ref{lem-cancel} on the above system, we get :
\begin{equation}\label{eqn-lem-split}
\sum_{j\in S}T_j' \in \ideal{f_i' \ |\ i\in\cI}=:I',
\text{ where, } f_i':=\frac{f_i}{\gcd(f_i,g_i)}, \forall i\in\cI.
\end{equation}
Since $f_i, g_i$ are not similar, $f_i'$ has degree $\geq1$, for all $i\in\cI$. 
Let the elements of $S$ be $s_1 < s_2 < \cdots < s_r$, for some $r\in[k-1]$. Since we have only
changed the non-nucleus part of $T_j$ to get $T_j'$, we deduce $K_{s_i}|T_{s_i}'$, for all $i\in[r]$.
Thus, modulo the ideal $I'':=\ideal{I',K_{s_1},\cdots,K_{s_{r-1}}}$, Equation (\ref{eqn-lem-split})
becomes : $T_{s_r}'\in I''$. 
We have $\radsp(I'')\subseteq\lrsp(\ell_i\ |\ i\in\cI)+K$.
Let us factor $T_{s_r}' = B_0B_1$, where $B_0$ is the product of all forms in $\radsp(I'')$
and $B_1$ is the remaining product. Thus, $B_0B_1\in I''$. By Lemma \ref{lem-non-zd}, $B_1$ can be 
cancelled out and we get $B_0\in I''$.

Suppose all forms of $B_0$ are in $K$, so $B_0 = K_{s_r}$. This means $K_{s_r} \in I''$ implying,
\begin{equation}\label{eqn-lem-split-2}
K_{s_r}\in\ideal{K_{s_1},\cdots,K_{s_{r-1}}, \{f_i' \ |\ i\in\cI\}}.
\end{equation}
Recall that each form in $f_i'$ is similar to some form in $(\FF^*\ell_i+K)$, for all $i\in\cI$.
Suppose form $(\beta_i\ell_i+u_i)|f_i'$, for all $i\in\cI$, for some $\beta$-s in $\FF^*$ and
$u$-s in $K$. In Equation (\ref{eqn-lem-split-2}) make the {\em evaluation} : $\ell_i\leftarrow$ 
$-\beta_i^{-1}u_i$, for all $i\in\cI$. This is a valid evaluation since 
$\{\ell_i \ |\ i\in\cI\}$ are linearly independent mod $K$, and values substituted are from $K$.
Clearly, this evaluation leaves the polynomial $K_s$ ($s\in S$) unchanged. Thus, we get 
$K_{s_r}\in\ideal{K_{s_1},\cdots,K_{s_{r-1}}}$, contradicting Claim \ref{clm-kmin}.

As a result, we have a form $\ell | B_0$ such that $\ell \notin K$.
We have $\ell \in \radsp(I'')$ $\subseteq\lrsp(\ell_i \ |\ i\in\cI)+K$, and 
by the way $T_{s_r}'$ was defined,
$\ell \notin \bigcup_{i\in\cI}(\FF^*\ell_i + K)$.
By the matching property of the nucleus, this gives us an $\ell' \in L^c_K(T_1)$ such that :
$\ell'\in(\lrsp(\ell_i  | i\in\cI)+K)\setminus K$ and 
$\ell'\notin \bigcup_{i\in\cI}(\FF^*\ell_i + K)$.
This means that there exist constants $\beta_i$-s in $\FF$, not all zero, such that
$\ell'\in \sum_{i\in\cI}\beta_i\ell_i+K$. As the coefficient of $y_0$ in $\ell'$ is nonzero while 
that in $\ell_i$ ($i\in\cI$) is $1$, we deduce :
$\trun(\ell') =(\sum_{i\in\cI} \beta_i)^{-1}(\sum_{i\in\cI}\beta_i\ell_i)$. 
If exactly one $\beta_i$ is nonzero, then $\ell'\in(\FF^*\ell_i + K)$, which is a contradiction.
So at least two $\beta_i$-s are nonzero, implying that $\trun(\ell')\in\trun(L^c_K(T_1))$
is a non-trivial combination of the $\ell_i$-s, contradicting the 
fact that $\{\ell_1,\cdots,\ell_{k-1}\}$ were obtained from $\sg_{k-1}(\trun(L^c_K(T_1)))$.

This contradiction proves that $S$ is split by $\pa(\ell_i')$, for some $i\in\cI$.
\end{proof}

To prove Theorem \ref{thm-strong-rank}, we need a combinatorial lemma
about general partitions. It is helpful to abstract out some of the 
details specific to identities and frame this as a purely combinatorial problem.
Since the proof is fairly involved, we present
that in the next subsection. For now, we give the necessary
definitions and claims. 
We have a universe $\cU := [k]$ of elements. We deal with set systems
with special properties.

\begin{definition}[Unbroken chain]\label{def-unbroken} 
A partition of $\cU$ is \emph{trivial} if it contains the single set $\cU$. 

Let $\bcP$ be a collection of non-trivial partitions of $\cU$ (here a collection refers to a {\em multiset}, 
i.e. $\bcP$ can have partitions repeated). A {\em chain in $\bcP$} is
a sequence of sets $A_1,A_2,\cdots,A_s$ (for some $s$) such that each set comes
from a different element of $\bcP$ (say $A_i \in \cP_i \in \bcP$).

The chain $A_1,A_2,\cdots,A_s$ is an \emph{unbroken chain}, if
$\overline{\bigcup_{i \in[s]} A_i}$ is non-empty and preserved in $\cP_i$, for each $i\in[s]$.
\end{definition}

Note that if $\overline{\bigcup_{i \leq s} A_i}$ is a singleton then
it is trivially preserved in any partition, therefore, such a chain would be unbroken.
By Lemma~\ref{lem-split}, the collection $\{\pa(\ell_i')|i\in[k-1]\}$ has
no unbroken chain. By purely studying partitions, we will show that such a phenomenon is absurd.  
The following combinatorial lemma
implies Theorem~\ref{thm-strong-rank}.

\begin{lemma}[Partitions have unbroken chain]\label{lem-unbroken} 
Let $\bcP$ be a collection of 
non-trivial partitions of $\cU$. If $\bcP$ contains at least $|\cU|-1$ partitions then $\bcP$ 
contains an unbroken chain.
\end{lemma}

\smallskip\noindent
{\bf Theorem \ref{thm-strong-rank}.} (Bound for simple, strongly minimal identities). 
Let $|\FF|>d$. The non-nucleus rank of 
a simple and strongly minimal $\sps(k,d)$ identity over $\FF$ is at most $\sg_{k-1}(\FF,d)$.
\begin{proof} (of Theorem~\ref{thm-strong-rank}) 
Let $C=\sum_{i\in[k]}T_i$ be a simple and strongly minimal $\sps(k,d)$ identity over $\FF$,
and let $K$ be the nucleus provided by Theorem \ref{thm-nucleus}. As $|\FF|>d$ we can assume (wlog by Lemma
\ref{lem-monic-forms}) the existence of a truncation operator on $L^c_K(T_1)$.
We will show that the rank of $\trun(L^c_K(T_1))$ is at most $\sg_{k-1}(\FF,d)$. 
By the matching property of the nucleus, $\trun(L^c_K(T_1))$ together
with $K$ span $L(C)$. Therefore, a non-nucleus rank bound of the former suffices
to bound the non-nucleus rank of $L(C)$.

Assuming that the rank of $\trun(L^c_K(T_1))$ 
is greater than $\sg_{k-1}(\FF,d)$, as in Lemma~\ref{lem-split},
we invoke $\sg_{k-1}(\trun(L^c_K(T_1)))$ to get $\{\ell_1,\ell_2,\cdots,\ell_{k-1}\}$.
Associated with each of these, we have the partition $\pa(\ell_i')$.
There are $k-1$ partitions in the collection $\bcP:=$ $\{\pa(\ell_i')|i\in[k-1]\}$, which are all non-trivial by the 
simplicity of $C$. Lemma~\ref{lem-unbroken} tells us that $\bcP$ has an unbroken chain, while
Lemma~\ref{lem-split} says that $\bcP$ has none. This contradiction implies the rank of 
$\trun(L^c_K(T_1))$ is at most $\sg_{k-1}(\FF,d)$, thus finishing the proof. 
\end{proof}

\subsubsection{The combinatorial proof of Lemma \ref{lem-unbroken}}

Intuitively, when the partitions in $\bcP$ have many classes then an unbroken chain should be easy to
find, for example, when $(k-1)$ partitions in $\bcP$ are all equal to $\{\{1\},\ldots,\{k\}\}$ then there is an easy
unbroken chain, namely $\{1\},\ldots,\{k-1\}$.  On the other hand, when the partitions in $\bcP$ contain
few classes then we can effectively decrease the universe and apply induction. Most of this subsection would deal
with the former case. Let us first define the splitting property.

\begin{definition}[Splitting property]
Let $\bcP$ be a collection of partitions of $\cU$. 
Suppose for all non-empty $S \subset \cU$, $S$ is split by at least $(|S|-1)$ partitions of $\bcP$.
Then $\bcP$ is said to have the \emph{splitting property}.
\end{definition}

\begin{claim} \label{clm-singleton} 
Let $\bcP$ be a collection of at least $(k-1)$ non-trivial partitions of $[k]$. 
If $\bcP$ has the splitting property then there is a chain $A_1,\cdots,A_{k-1}$ in $\bcP$ such that
$\overline{\bigcup_{i \leq k-1} A_i} = \{k\}$. (In particular, $\bcP$ has an unbroken chain.)
\end{claim}

We defer its proof and, instead, first show why this claim would suffice.

\begin{proof} (of Lemma~\ref{lem-unbroken}) 
We will prove this by induction
on the universe size $k$. For the base case, suppose $k=3$ and $\bcP=$ $\{\cP_1,\cP_2,\ldots\}$. 
So we have at least two partitions. If any partition (say $\cP_1$) contains exactly $2$ sets,
it must be a pair and a singleton (say $\cP_1=\{\{1,2\},\{3\}\}$). 
But then $\{1,2\}$ is itself an unbroken chain in $\bcP$. So, all the partitions
can be assumed to consist only of singletons. But then we can take the set, say, $\{1\}$
from $\cP_1$ and, say, $\{2\}$ from $\cP_2$ to get an unbroken chain.

Now for the induction step. Suppose $\bcP$ has at least $(k-1)$ partitions.
We assume that the claim is true for
universes of size upto $(k-1)$. If $\bcP$ has the splitting property,
then we are done by Claim~\ref{clm-singleton}. If not, then for some subset $S \subset\cU$ of
size at least $2$, $S$ is split in at most $(|S|-2)$ partitions. Let the collection of partitions
in $\bcP$ that preserve $S$ be ${\bf\bcP'}$. So ${\bf \bcP'}$ contains
at least $(k-1)-(|S|-2)$ $=(k-|S|+1)$ partitions. Merge the elements
of $S$ into a new element, to get a new universe $\cU'$ of size $(k-|S|+1)$. 
The partitions in ${\bf \bcP'}$ are valid partitions of $\cU'$, and still maintain their structure.
We now have a universe of size $k-|S|+1<k$, and at least $k-|S|+1$ partitions.
By the induction hypothesis, there is an unbroken chain in ${\bf \bcP'}$.
Observe that it is (under the natural correspondence) still an unbroken chain in the original collection
$\bcP$, and we are done.
\end{proof}

\begin{proof} (of Claim~\ref{clm-singleton})
We will label the partitions in $\bcP$ in such a way that its first $(k-1)$ elements, 
$\cP_1,\cdots,\cP_{k-1}$ satisfy : $\cP_{i}$ splits $\{i,k\}$, for all $i\in[k-1]$. Thus, 
there is a set $A_{i} \in \cP_{i}$
that contains $i$ but not $k$. Naturally, $\overline{\bigcup_{i \leq k} A_i} = \{k\}$.

We will construct this labelling through an iterative process.
In the $i$th phase, we will 
find $\cP_i$. At the end of this phase, we will
have $\cP_1,\cdots,\cP_{i}$ with the desired property 
and the remaining {\em pool} $\bcP$ of remaining partitions. 
We warn the reader that this labelling is very dynamic, so
during the $i$th phase, we may change the labels
of $\cP_1,\cdots,\cP_{i-1}$ by moving them to $\bcP$
and labelling new partitions with older labels.
At any stage, we have the \emph{labelled}
partitions and the {\em unlabelled} partitions $\bcP$.
Before the beginning of the first phase, $\bcP$ is just the given collection
of all permutations.

{\bf [Phase $1$]}
The first phase is easy to understand.
By the splitting property,
there is some partition that splits $\{1,k\}$.
We set this to $\cP_1$. 

{\bf [Phase $i$]}
The $i$th phase, $i\geq2$, is a rather involved process. We describe the
various sets associated with it and explain them. By the beginning of this phase,
we have already determined $\cP_1,\cdots,\cP_{i-1}$. The \emph{covered}
elements are just $[i-1]$. We maintain
a partition $E_1,\cdots,E_{i-1}$ of the covered elements. 
We set $E_0 = \{i,k\}$.
Corresponding to each set $E_j$, we have a set
of partitions $\cC_j$ ($:=\{\cP_b|b\in E_j\}$). We fix $\cC_0=\emptyset$.
Note that the $\cC_j$'s form a partition of the labelled partitions.
We set $E_{\leq j} = \bigcup_{0 \leq l \leq j} E_l$.
We get a similar set of partitions $\cC_{\leq j} = \bigcup_{0 \leq l \leq j} \cC_l$.
We will \emph{always} maintain that $\cP_j$ splits
$\{j,k\}$.

There will be various rounds in a phase.
To aid understanding, we will describe the first 
and second round in detail.

{\bf [Phase $i$, Round $1$]}
We now explain the first round.
Initially, we set $E_1 = [i-1]$. 
This is because, as of now, we know nothing
about the elements in the set containing $k$ in
the various $\cP_j$'s. 
Note that \emph{at any stage}, if we have
a partition in $\bcP$ that splits $\{i,k\}$,
we can set this to $\cP_i$ and we are done.

Now, for every element $b$ in $E_1$ (currently, it is $[i-1]$)
check if there is a partition in $\bcP$ that
splits $\{i,k,b\}$. If so, call this a {\em success} for $b$. 
Note that $\{i,k\}$ is not split in any partition $\bcP$.
Now, we can label \emph{this}
partition as $\cP_b$ and move the old one to the pool $\bcP$.
All the labelled partitions still have their desired property.
If the old partition splits $\{i,k\}$, then we
are done (since this is now in $\bcP$). So, we can assume
that (even after this switching) that $\{i,k\}$ is preserved
in all of $\bcP$. 
For our (new) $\cP_b$, we know that $\{i,k\}$
is preserved. So we have some extra information about it.
This is represented by ``promoting" $b$ from $E_1$ to $E_2$.
This just involves removing $b$ from $E_1$ and putting it in $E_2$.
Let us repeat this for all elements in $E_1$ until
we have a maximal set $E_2$, containing all successful
elements. Note that when $\bcP$ changes because of the
switching, we check all elements in $E_1$
again for successes.

We are now at the end of this round and have the following information.
$\{E_1, E_2\}$ is a partition of $[i-1]$. 
For any successful element $b \in E_2$, $\cP_b$ preserves
$\{i,k\}$. So, the labelled partitions $\cC_2$ preserve
$\{i,k\} = E_0$. 
For every failure $b \in E_1$, \emph{every} partition in $\bcP$
preserves $\{i,k,b\}$. In other words,
every partition in $\bcP$ preserves $E_1 \cup \{i,k\} = E_{\leq 1}$. 
Successes create the new $E_2$, while failures increase the size of the
set preserved by $\bcP$. 

Let us understand this a little more.
Suppose \emph{all} elements are eventually successful, so $E_2 = [i-1]$.
Therefore, all labelled partitions preserve $\{i,k\}$. But so
do all partitions in $\bcP$.
So $\{i,k\}$ is preserved in \emph{all} partitions,
contradicting the splitting property. There must
be some failures. 
Suppose everything is a failure, so $E_1$ is still $[i-1]$. 
The set $E_{\leq 1}$ has size $i+1$. But the only partitions that split $E_{\leq 1}$
are the labelled ones since $\bcP$ preserves $E_{\leq 1}$.
There are only $i-1$ labelled partitions so this contradicts
the splitting property. So there are some successes and some 
failures and $E_1, E_2$ form a non-trivial partition. In some
sense, we made ``progress".

{\bf [Phase $i$, Round $2$]}
We move to the next round. For every $b \in E_2$, we check
if $E_{\leq 1} \cup \{b\}$ is split in any partition of $\bcP$. If we get
a success, then we set this partition to be the new $\cP_b$. 
We ``promote" $b$ from $E_2$ to a new set $E_3$.
We need to shift this old partition (call it $\cP$) 
to our pool $\bcP$. But we want to ensure that $E_{\leq 1}$ is preserved in all of $\bcP$,
and this may not happen for $\cP$. So, first we check if $\{i,k\}$
is preserved in $\cP$. If not, we are done. Assume otherwise. We start
checking if $\{i,k,c\}$ is preserved, for all $c \in E_1$.
If it is so for all $c$, then we know that $E_{\leq 1}$ is preserved in
$\cP$. So, we maintain our condition about $\bcP$, and we continue
to the next $b$. If not (this is the interesting part!),
then we have found a partition that separates $c$ from $\{i,k\}$. 
Note that the reason why $c$ belongs to $E_1$, is, because we were unable
(in the previous round) to find such a partition. So, we label
$\cP$ as $\cP_c$. We ``promote" $c$ from $E_1$ to $E_2$.
The old $\cP_c$ is moved to the pool $\bcP$, so we repeat the 
above procedure for this partition as well. So, either we maintain
the invariant that $E_{\leq 1}$ is preserved in all of $\bcP$, or we promote
elements from $E_1$ to $E_2$. If, at some stage, there are no elements
in $E_1$, then we are done. (Why? Because every labelled partition now
preserves $\{i,k\}$, by the splitting property, there must be a partition in $\bcP$ splitting
this.) For all the failures $b \in E_1$, we know that $E_{\leq 1} \cup \{b\}$ is preserved
in $\bcP$. All successes are promoted to $E_2$. So at the end we have
the partition $E_1, E_2, E_3$ of $[i-1]$. All of $\bcP$
preserves $E_{\leq 2}$. The partitions $\cC_3$ preserve $E_{\leq 1}$ and
those in $\cC_2$ preserve $E_{\leq 0}$. If $E_3$ is empty,
then $E_{\leq 2}$ is of size $i+1$.
All of $\bcP$ preserves $E_{\leq 2}$ so the splitting property
is violated. If $E_1$ is empty, then also
we are done.

Let us give a formal proof by describing the invariant at the end
of a round.

\begin{claim} \label{clm-round} By the end of the ($i$th phase-) $j$th round, 
suppose we do not find the right $\cP_i$.
Then we can construct a partition of $[i-1]$, $E_1,\cdots,E_{j+1}$,
where $E_{j+1}$ is non-empty and the following hold:
the partitions $\cC_l$ preserve $E_{\leq l-2}$, for all $2\le l\le(j+1)$, and the unlabelled
partitions $\bcP$ preserve $E_{\leq j}$.
\end{claim}
\begin{proof} We prove by induction on $j$. We have already proven this 
for $j=1,2$. Assuming this is true upto $j$, we will show this for $j+1$.
The round repeatedly ``processes" elements of $E_{j+1}$.
Processing $b \in E_{j+1}$ involves checking if all
partitions in $\bcP$ preserve $E_{\leq j} \cup \{b\}$. If they do, then $b$
is a failure. If $\cP \in \bcP$ splits $E_{\leq j} \cup \{b\}$, then
we ``swap" it with $\cP_b$, i.e. $\cP$ is now the old $\cP_b$ and is
denoted a \emph{hanging partition}. The element $b$ is promoted
from $E_{j+1}$ to the new set $E_{j+2}$. How to deal with the hanging partition $\cP$?
We first check if it splits $E_0$. If so, we have found $\cP_i$. Otherwise,
we check if it preserves $E_0 \cup \{c\}$, for all $c \in E_1$.
If it splits $\{i,k,c\}$, then we swap $\cP$ with $\cP_c$. We promote
$c$ from $E_1$ to $E_2$. The old $\cP_c$ becomes the new hanging partition $\cP$.
If $\cP$ preserves $E_{\leq 1}$, then we move on to $E_2$.
In general, if $\cP$ preserves $E_{\leq l}$,
then we check if $\cP$ preserves all $E_{\leq l} \cup \{c\}$, for $c \in E_{l+1}$.
If $\cP$ splits $E_{\leq l} \cup \{c\}$, we swap $\cP$ with $\cP_c$ and
promote $c$ from $E_{l+1}$ to $E_{l+2}$. Note that the sets $E_{\leq p}$ (for any $p$)
can only decrease on such a promotion. So still the partitions
in $\cC_p$ preserve $E_{\leq p-2}$.
The old $\cP_c$ becomes the new hanging partition $\cP$
and we repeat this process. If, on the other hand, $\cP$ preserves all $E_{\leq l} \cup \{c\}$,
then $\cP$ preserves $E_{\leq l+1}$. So we repeat this process
with $E_{l+2}$, and so on. If we end up with $\cP$
preserving $E_{\leq j}$, then
we can safely move $\cP$ into $\bcP$. Otherwise, we have made a promotion and
we deal with a new hanging partition.
Note that when $\bcP$ changes, we again process all elements in $E_{j+1}$.
There can only be a finite number of promotions,
so this round must end. We end up with $E_0,\cdots,E_{j+2}$,
with $\cC_l$ preserving $E_{\leq l-2}$.
All the failures are still in $E_{j+1}$, and $\bcP$
preserves all $E_{\leq j} \cup \{c\}$, $\forall c \in E_{j+1}$.
So $\bcP$ now preserves $E_{\leq j} \cup E_{j+1} = E_{\leq j+1}$. 
Note that if $E_{j+2}$ is empty,
we have a contradiction. This is because $E_{\leq j+1}$
is of size $i+1$ and there are at most $i-1$ partitions splitting it.
\end{proof}

Now we show that in this phase $i\geq2$ there can be at most $i$ rounds before
we get the desired $\cP_i$.

\begin{claim} \label{clm-nonempty} Suppose $E_1,\cdots,E_j$
is a partition of $[i-1]$ such that $C_l$ preserves $E_{\leq l-2}$
and $\bcP$ preserves $E_{\leq j}$. Then all $E_l$'s are non-empty.
\end{claim}
\begin{proof} Suppose $E_l$ is empty, for some $l\in[j]$. So $C_l$ is also
empty. Any partition that is not in $C_{\leq l-1}$
is either in $\bcP$ or in $C_p$, for some $p \geqq l+1$ (if it exists). All these 
partitions preserve $E_{\leq l-1}$. Thus, the only partitions splitting $E_{\leq l-1}$
are those in $C_{\leq l-1}$.
Since $|C_{\leq l-1}| = |E_{\leq l-1}| - 2$, we contradict the splitting property.
\end{proof}

The sets $E_1,\cdots,E_j$ form a partition of $[i-1]$. The above
claim tells us that we can run at most $i-1$ rounds to completion.
Hence, if we do not find $\cP_i$ by $i-1$ rounds, then, by Claim~\ref{clm-round},
we will find it in the $i$th round. This completes the proof.
\end{proof}

\subsection{The general case}\label{sec-gen-case}

Now, we deal with simple, minimal identities and remove
the strong minimality condition. This will come at a cost
of an extra $k$ factor in the rank bound. First, we recall the definition of gcd and simple
parts of a general $\sps$ circuit, as given in older works \cite{DS06,SS09}.

\begin{definition}[Gcd \& Simple part]
Let $C=\sum_{i\in[k]}T_i$ be a $\sps(k,d)$ circuit over a field $\FF$. The {\em gcd} of $C$ is
defined to be the usual gcd of the polynomials $T_i$-s, i.e. $\gcd(C):=\gcd(T_i|i\in[k])$.

The {\em simple} part of $C$ is the $\sps(k,d')$  circuit, $\simp(C):=C/\gcd(C)$, where
$d':=d-\degr(\gcd(C))$. 
\end{definition}

The following will be shown to be a consequence of Theorem~\ref{thm-strong-rank}.

\smallskip\noindent
{\bf Theorem \ref{thm-rank}.} (Final bound). Let $|\FF|>d$. The rank of a simple, minimal
$\sps(k,d)$, independent-fanin $k'$, identity is at most $2k^2 + (k-k')\cdot\sg_{k'}(\FF,d)$.
\begin{proof} 
Let circuit $C$ be $T_1+\cdots+T_k = 0$. Wlog let $T_1,\cdots,T_{k'}$ be a 
linear basis for $T_1,\cdots,T_k$. Obviously, we have $1< k' <k$ (first by simplicity and second 
by zeroness). By Theorem \ref{thm-nucleus}, there exists a nucleus $K$ wrt the set $\cI:=[k']$. The rank 
of $K$ is at most $2k^2$. 
So, it remains to bound the non-nucleus rank of $C$ by $(k-k')\cdot\sg_{k'}(\FF,d)$.

As $T_1,\cdots,T_{k'}$ form a basis, for each $i\in[k'+1,k]$, there exists $\alpha_{i,j}$-s in
$\FF$ such that we have a zero circuit
$D_i := \sum_{j\in[k']}\alpha_{i,j}T_j + T_i = 0$. Define $N_i$ to be the set of $j$-s for which
$\alpha_{i,j}\ne0$. Thus,
\begin{equation}\label{eqn-thm-rank}
\forall i\in[k'+1,k],\ D_i = \sum_{j\in N_i}\alpha_{i,j}T_j + T_i = 0
\end{equation}
Since $\{\alpha_{i,j}T_j\ |\ j\in N_i\}$ are $|N_i|$ linearly independent terms, we get
that $D_i$ is a strongly minimal $\sps(|N_i|+1,d)$ identity, for all $i\in[k'+1,k]$. 
By nucleus properties, $\{K_j|j\in N_i\}$ are linearly independent polynomials, implying that 
the polynomials $\{K_j/g_i|j\in N_i\}$ are also linearly independent,   
where $g_i:=M(L_K(\gcd(D_i)))$. Thus, the linear space $K$ remains a nucleus of the new identity 
$\simp(D_i)$, showing at the same time that it is strongly minimal. We conclude that $\simp(D_i)$
is a simple, strongly minimal $\sps(k_i,d_i)$ identity with nucleus $K$ (although of $\rk<2k^2$), 
$k_i\le(k'+1), d_i\le d$,
for all $i\in[k'+1,k]$. Theorem \ref{thm-strong-rank} bounds the non-nucleus (non-$K$ to be precise) 
rank of each of these identities by $\sg_{k'}(\FF,d)$.

Suppose a linear form $\ell|\gcd(D_i)$ for all $i\in[k'+1,k]$. Then $\ell$ divides $T_j$
for all $j\in\bigcup_{i\in[k'+1,k]}N_i$ $\cup[k'+1,k]$. Consider the case 
$\bigcup_{i\in[k'+1,k]}N_i=[k']$, it means that $\ell$ divides every term in $C$, 
contradicting simplicity. Thus, in that case every linear form $\ell$ of $C$ appears in at least 
one of the circuits $\{\simp(D_i)|i\in[k'+1,k]\}$, whose total non-nucleus rank we have already bounded
by $(k-k')\cdot\sg_{k'}(\FF,d)$, so we will be done.

The case, left to handle, is when : $S:=\bigcup_{i\in[k'+1,k]}N_i\subsetneq[k']$. This means,
by summing over $i$ in Equation (\ref{eqn-thm-rank}), $\sum_{i\in[k'+1,k]}T_i=$ 
$\sum_{s\in S}\beta_sT_s$, for some $\beta$-s in $\FF$.
Substituting this in the equation $C=0$ we get, 
$$C=C_{[k']}+C_{[k'+1,k]}\ =\ \sum_{i\in[k']}T_i+\sum_{s\in S}\beta_sT_s\ =\ 0.$$
As $S$ is a proper subset of $[k']$, the above equation could only mean that a nontrivial 
combination of $T_i$ $(i\in[k'])$ is vanishing, contradicting the linear independence of those
polynomials. Thus, $S=[k']$. This completes the proof.
\end{proof}

\section{Sylvester-Gallai Rank Bounds for any $\FF$} \label{sec-sg}

We wish to bound $\sgkb{\FF,m}$, for \emph{any} field $\FF$. We will prove the
following theorem, which can be seen as the first attempt ever to give a Sylvester-Gallai Theorem for 
{\em all} fields. It is convenient to think of a set of vectors $S$ in $\FF^n$ as 
\emph{multiple-free}: this means that no two vectors
in $S$ are scalar multiples of each other.

\begin{theorem}[High dimension Sylvester-Gallai for any field]\label{thm-sgkbd}
Suppose $k \in {\NN}^{>1}$ and $S$ is an $SG_k$-closed set of vectors in $\FF^n$
of rank $r \ge 9k$. Then, $|S| \ge 2^{r/9k}$. In other words,  for every $m\in\NN^{>1}$,
$SG_k(\FF,m)\leq 9k\lg m$.
\end{theorem}

\noindent {\bf Remark.} This bound is not tight. Over $\FF_p$ , the best construction
we can come up with is an $SG_k$-closed set 
with rank around $k+\log_p m$. Consider $\FF_p^{k-1+r}$, for some $r > 0$ and $p\nmid(k-1)$.
Take the set $S_1\subset\FF_p^{k-1+r}$ of vectors $e_1,e_2,\cdots,e_{k-1}$,
$\frac{1}{k-1}\cdot\sum_{i\in[k-1]} e_i$,
where $e_i$ is just the unit vector in the $i$th direction. Then take
the set $S_2\subset\FF_p^{k-1+r}$ of all (non-zero) vectors which have zeroes in the first
$(k-1)$ coordinates. Extend this to $\FF_p^{k+r}$ by putting a $1$
in the new coordinate. This gives a set of non-multiple
distinct vectors in $\FF_p^{k+r}$.
Observe that $S_1$ is $SG_{k-1}$-closed 
and $S_2$ is $SG_2$-closed. Hence $S_1 \cup S_2$ is 
$SG_k$-closed. The size of $S$ is $(k + p^r-1)$ and the rank
is just $k+r$.\\


In some sense, bounds for $\sg_2(\FF,m)$ are already implicit in known theorems
(used to prove lower bounds for LDCs). Concretely, Corollary 2.9 of~\cite{DS06}
can be interpreted as a proof that $\sg_2(\FF,m) = O(\log m)$. This is an extension
of theorems in~\cite{GKST02} that prove this for $\FF_2$. In the context
of $\sg_2$, these proofs
can be interpreted as a ``doubling trick". In essence, each time we
want to increase the rank of an $\sg_2$-closed set by $1$, we are forced
to double the number of vectors. A naive attempt to implement this
for $\sg_k$ does not work. 

Roughly speaking, we want to argue that if we want to increase the rank
of an $\sg_k$-closed set by $(k-1)$, then the size of the set must double.
But, when $k\geq3$, this is not true! It is possible to increase the rank by $(k-1)$
by adding a very small number of vectors. So we have a sort of two-pronged approach.
If the size does not increase much, even though the rank increases,
then we show that the set has some very special $\sg$ properties. Namely, {\em many} small
subspaces of the set are $\sg_{k-1}$-closed. Even though these subspaces
can intersect very heavily, we are still able to argue that the set
must now be very large. 

We will require two auxiliary claims. The first claim is probably of independent interest,
but the second is tied to our current approach.

\begin{claim} \label{clm-coord} Suppose $k \in {\NN}^{>1}$ and $S$ is $SG_k$-closed.
Let vectors $e_1,e_2,\cdots,e_r$ be elements of $S$ that form a basis for $S$. Hence, every element in $S$
is represented by an $r$-tuple of coordinates in $\FF$. There exists some element
in $S$ whose representation has at least $r/(k-1)$ non-zero coordinates.
\end{claim}

\begin{proof} Consider any vector $v \in S$. Let $v = \sum_{i=1}^r \alpha_i e_i$, for $\alpha_i \in \FF$.
We denote by $N(v)$ the index set $\{i | \alpha_i \neq 0\}$. In other words,
$N(v)$ is the set of indices for which the corresponding coordinates of $v$
are non-zero. Let $v_1 \in S$ be some vector that maximizes $|N(v_1)|$. Choose some $v_2 \in S$
such that $N(v_2) \cap N(v_1) = \emptyset$ and $|N(v_2)|$ is maximized. Such a $v_2$ exists as long as
$|N(v_1)|<r$, as in this case we can get $v_2$ from $\lrsp(e_i|i\notin N(v_1))$. 
Iteratively, choose $v_j \in S$ such that $N(v_j)$ is disjoint to $\bigcup_{l < j} N(v_l)$
and $|N(v_j)|$ is maximized. As long as $\bigcup_{l \leq j} N(v_l) \neq [r]$,
we can always choose the next $v_{j+1}$, again, from $\lrsp(e_i|i\notin\bigcup_{l \leq j} N(v_l))$. 
We keep choosing $v_j$'s until we cover all
coordinates. At the end, $\bigcup_{l \leq j} N(v_j) = [r]$. Note that the
sets $N(v_1),\cdots,N(v_j)$ form a partition of $[r]$.

Suppose $j \geq k$. Then take the vectors $v_1,\cdots,v_k$. They are certainly
linearly independent, since they are defined on a disjoint set of coordinates. By the $SG_k$-closure
of $S$, some non-trivial linear combination of these vectors exists in $S$. Suppose some
non-zero combination of $v_{i_1},v_{i_2},\cdots$, denoted by $v$, is in $S$ (where $i_1 < i_2 \cdots$).
Note that $N(v) \supset N(v_{i_1})$ and $N(v)$ is disjoint to $\bigcup_{l < i_1} N(v_l)$.
That contradicts that choice of $v_{i_1}$. Hence $j < k$. Because the sets
$N(v_l)$ form a partition of $[r]$, $|N(v_1)|\geq r/(k-1)$. That completes the proof.
\end{proof}

\begin{claim} \label{clm-subsets} Suppose $k \in {\NN}^{>2}$. 
Consider a set of linearly independent vectors
$e_1,\cdots,e_{r'}$ in $S$. Let $I \subseteq [r']$ and $|I|\geq r'/4$. Let the set $E_I$ be $\{e_i | i \in I\}$ 
and $S_I$ be the set $S \cap \lrsp(E_I)$.
If for all such $I$, $S_I$ is $SG_{k-1}$-closed, then $|S_{[r']}| \geq 2^{r'/8k}$.
\end{claim}
\begin{proof} 
By Claim~\ref{clm-coord}, for every such $I$, 
there is a $v_I \in S$
such that $v_I$ has at least $|I|/(k-2)$ non-zero coordinates (wrt basis $E_I$). Fix these $v_I$, for
each such $I$. As before, we denote the set of indices corresponding to non-zero coordinates of $v$ 
by $N(v)$. We describe a random
process to generate a subset of $[r']$. Simply pick each element in $[r']$
independently with probability $1/2$. Let $I$ and $J$ be two sets 
generated independently this way. A ``good" event occurs if $|I|\geq r'/4$ and
$v_I \notin S_J$. We will call this the {\em good event for $(I,J)$}.
How can a good event not happen? Either 
$|I|< r'/4$, or, if $|I|\geq r'/4$, then $v_I \in S_J$.
What is the probability of a ``bad" event (call this $\cE(I,J)$) happening?
This is:
$$ \pr[|I|< r'/4]\ +\ \pr[|I|\geq r'/4]\cdot\pr[v_I \in S_J \ \big| \ |I|\geq r'/4]  $$
Since $\EX[|I|] = r'/2$, the probability that $|I| < r'/4$ is at most
$e^{-r'/8}$ (by a Chernoff bound, refer to notes \cite{oDCher} for the exact form used). 
Obviously, $\pr[|I|\geq r'/4] < 1$.
Now assume that $|I|\geq r'/4$. So $|N(v_I)| \geq r'/4(k-2)$. For $v_I$ to
be in $S_J$, $J$ must contain $N(v_I)$. By the random construction of $J$,
the probability of this is at most $2^{-r'/4(k-2)}$. Thus, the probability
of a bad event is at most $e^{-r'/8} + 2^{-r'/4(k-2)} \leq 2\cdot2^{-r'/4(k-1)}$.

Now, let us choose $q = 2^{r'/8k}$ subsets of $[r']$ independently
through this random distribution. Call these $I_1,I_2,\cdots,I_q$.
For indices $1\le a< b \leq q$, let $X_{ab}$ be the indicator random variable 
for the event $\cE(I_a,I_b)$. What is the expected total number
of bad events?
$$ \EX\left[\sum_{1\le a< b \leq q} X_{ab}\right] = \sum_{1\le a< b \leq q} \EX[X_{ab}]
< \frac{q^2}{2} \pr[\cE(I_1,I_2)] \leq 2^{r'/4k} \cdot 2^{-r'/4(k-1)} < 1 $$
By the probabilistic method, there exist subsets of $[r']$,
$I_1,I_2,\cdots,I_q$ such that no event $\cE(I_a,I_b)$ happens.
This means all sets $I_a$ are of size at least $r'/4$.
Furthermore, for $1\le a< b \leq q$, $v_{I_a} \notin S_{I_b}$. This
means that $v_{I_a} \neq v_{I_b}$. Therefore, there must
be at least $q$ distinct vectors in $S_{[r']}$.
\end{proof}

\begin{proof} (of Theorem~\ref{thm-sgkbd}) It will be convenient to assume that $|S| < 2^{r/9k}$,
and arrive at a contradiction. Let $T \subset S$ be a subset of rank $t\leq r/9$. 
We can construct a basis of $S$, $e_1,\cdots,e_{r'},e_{r'+1},\cdots,e_r$, using elements of $S$ 
such that :
$r' \geq 8r/9$, $\{e_1,\cdots,e_{r'}\}$ are orthogonal to $T$, and $\{e_{r'+1},\cdots,e_r\}$
spans $T$. If $k>2$ then by applying 
Claim~\ref{clm-subsets} to $e_1,\cdots,e_{r'}$ we get some 
$I\subseteq[r']$ such that $|I|\geq r'/4\geq$ $2r/9\geq 2k$, and $S_I$ is \emph{not} $SG_{k-1}$-closed 
(recall, $S_I:=S\cap\lrsp(e_i|i\in I)$).
There exist linearly independent $v_1,v_2,\cdots,v_{k-1} \in S_I$ such that no non-trivial
combination of these are in $S_I$ (and hence in $S$). On the other hand, if $k=2$ then define $v_1:=e_1$. 
Note that $v_1,v_2,\cdots,v_{k-1}$ are all orthogonal to $T$.

Consider any $v \in T$. By the $SG_k$-closure of $S$,
there exists some linear combination of $\{v,v_1,\cdots,v_{k-1}\}$ in $S$.
Call this the \emph{image} of $v$. Take two different $v,v' \in T$.
We argue that their images are distinct. Let the image of $v$ 
be $\alpha v + \sum_{i \leq k-1} \alpha_i v_i$ and the image
of $v'$ be $\beta v' + \sum_{i \leq k-1} \beta_i v_i$. 
Note that $\alpha\beta\ne0$ by the way we have chosen $v_1,v_2,\cdots,v_{k-1}$.
If these images are equal, then $\alpha v - \beta v' = \sum_{i \leq k-1} (\beta_i - \alpha_i)v_i$.
Since $v$ and $v'$ are not multiples of each other, the right hand side
is non-zero. The left hand side is a vector that is spanned by $e_{r'+1},\cdots,e_r$
but the right hand side is spanned by $e_1,\cdots,e_{r'}$. Contradiction.
Hence, all images are distinct.

Starting from any subset $T \subset S$ of rank $t\leq r/9$,
there exist linearly independent $v_1,\cdots,v_{k-1}$ $\in S$ such that
$S \cap \lrsp\{v_1,\cdots,v_{k-1},e_{r'+1},\cdots,e_r\}$
contains at least $2|T|$ vectors. The rank of this intersection is $t + (k-1)$.
Starting from $T$ being just one vector, we keep repeating this process.
This can go on for $h$ iterations, where $h$ is the smallest integer
such that $1 + (k-1)h \geq r/9$. So, $h = \lceil (r/9 - 1)/(k-1) \rceil$ which is at least $r/9k$, since
$r\geq 9k$.
We have $|S| \geq 2^h \geq 2^{r/9k}$. Contradiction.
\end{proof}

\section{Conclusion}

In this work we developed the strongest methods, to date, to study depth-$3$ identities.
%
The ideal methods hinge on a classification of zerodivisors of the ideals
generated by gates of a $\sps$ circuit (eg. Lemmas \ref{lem-non-zd}, \ref{lem-I-match} and 
\ref{lem-cancel}). That is useful in proving an ideal version of Chinese remaindering tailor-made
for $\sps$ circuits, which is in turn useful to show a connection between all the gates involved
in an identity. As a byproduct, it shows the existence of a low rank {\em nucleus identity} $C'$ 
sitting inside {\em any} given $\sps(k,d)$ identity $C$ (when $C$ is not minimal, $C'$ can still
be defined but it might not be homogeneous). This nucleus identity is quite mysterious and it
might be useful for PIT to understand (or classify) it further. For example, {\bf can the 
rank bound for the nucleus identity be improved to $O(k)$?} 

We generalize the notion of Sylvester-Gallai configurations to {\em any} field
and define a parameter $\sg_k(\FF,m)$ associated with field $\FF$. This number seems to be a
fundamental property of a field, and as we show, is very closely related to $\sps$ identities. 
It would be interesting to obtain bounds for $\sg_k(\FF,m)$ for different $\FF$. For example,
as also asked by \cite{KSar09}, {\bf can we nontrivially bound the number $\sg_k(\FF,m)$ for interesting 
fields: $\CC$, 
finite fields with large characteristic, or even $p$-adic fields?}
Other than the bounds for $\RR$, all that was known before
is $\sg_2(\CC,m)\leq 3$ \cite{EPS03}. 
We shed (a little) light on high dimension $\sg$ rank bounds by showing $\sg_k(\FF,m) = O(k\log m)$.
We conjecture: $\sg_k(\FF,m)$
is $O(k)$ for zero characteristic fields, while $O(k+\log_pm)$ for fields of characteristic $p>1$.

We also prove a property of a general collection $\bcP$ of partitions of a universe $\cU$, namely,
if $\bcP$ has at least $|\cU|-1$ partitions then it has an unbroken chain. It is tight and gives an idea 
of how a Sylvester-Gallai configuration in the non-nucleus part of a $\sps(k,d)$ identity ``spreads'' around. 

Finally, we ask: {\bf Can the rank bound for simple minimal real $\sps(k,d)$ identities be improved to 
$O(k)$?} The best constructions known, since \cite{DS06}, have rank $4(k-2)$. Likewise, over other fields,
our upper bound of $O(k^2\log d)$ still leaves some gap in understanding the exact dependence on $k$.

\section*{Acknowledgements}

We are grateful to Hausdorff Center for Mathematics, Bonn
for the kind support, especially, hosting the second author  
when part of the work was done. The first author thanks
Nils Frohberg for several detailed discussions that clarified 
the topic of incidence geometry and Sylvester-Gallai theorems.
We also thank Malte Beecken, Johannes Mittmann and Thomas Thierauf
for several interesting discussions.

\bibliographystyle{alpha}
\bibliography{pit-refs}

\appendix
  
\section{Technical, Algebraic Lemmas}

We denote the polynomial ring $\FF[x_1,\ldots,x_n]$ by $R$. 

\begin{lemma}[Monic forms]\label{lem-monic-forms}
Let $|\FF|>d$ and $C$ be a $\sps(k,d)$ identity, over $\FF$, with nucleus $K$.
Let $y_0\in L(R)^*$ and $U$
be a subspace of $L(R)$ such that $L(R)=\FF y_0\oplus U\oplus K$.
Then there exists an invertible linear transformation $\tau:L(R)\rightarrow L(R)$
that fixes $K$ and :
\begin{itemize}
\item[1)] $\tau(C)$ is also a $\sps(k,d)$ identity with nucleus $K$ and the same simplicity,
minimality properties. 
\item[2)] Every form in $L^c_K(\tau(C))=$ $\tau(L^c_K(C))$ is monic wrt $y_0$.
\end{itemize}
\end{lemma}
\begin{proof}
Let $r:=\rk(\FF y_0\oplus U)$.
Fix a basis $\{y_0,\ldots,y_{r-1}\}$ of $\FF y_0\oplus U$ and let $\vec{y}$ denote the column 
vector $[y_0 \cdots y_{r-1}]^T$.  
Let $\ell\in L^c_K(T_1)$.  Then there is a unique nonzero (column) vector $\vec{\alpha_\ell}\in\FF^r$ 
and a $v_\ell\in K$,
such that $\ell=\vec{\alpha_\ell}^T\cdot\vec{y}$ $+v_\ell$. We intend $\tau$ to be a linear transformation 
that fixes each element in $K$ and maps $\vec{y}$ to $A\vec{y}$ where $A\in\FF^{r\times r}$.
Such a $\tau$ will map $\ell$ to $\tau(\vec{\alpha_\ell}^T\cdot\vec{y})+v_\ell=$ 
$\vec{\alpha_\ell}^T\cdot\tau(\vec{y})+v_\ell=$ $\vec{\alpha_\ell}^T A\vec{y}+v_\ell$. 
To make $\tau(\ell)$ monic in 
$y_0$ we need to choose $A$ such that the first
coordinate in $\vec{\alpha_\ell}^T A$ is nonzero, i.e. $\vec{\alpha_\ell}^T A_{*1}\ne0$ where $A_{*1}$ 
is the first column of $A$. Thus, we want an $A$ such that 
$\prod_{\ell\in L^c_K(T_1)}\vec{\alpha_\ell}^T A_{*1}\ne0$. 

Now the nonzero multivariate polynomial $f(\vec{Y}):=$ $\prod_{\ell\in L^c_K(T_1)}\vec{\alpha_\ell}^T\vec{Y}$ 
has degree at most $d<|\FF|$. Hence, by the Schwartz-Zippel lemma \cite{Sch80, Z79} there exists a point 
$\vec{Y}\in\FF^r$ at which $f$ is nonzero. We can fix $A_{*1}$ to be that point. 
This fixes just one column of $A$ to a nonzero vector and we can arbitrarily fix the rest such that $A$ is an
invertible matrix. Thus, the corresponding invertible $\tau$ makes each $\ell\in L^c_K(T_1)$ monic in $y_0$. 
Since
$\tau$ fixes the nucleus $K$, matching property of the nucleus tells us that every form in $L^c_K(\tau(C))=$ 
$\tau(L^c_K(C))$ is monic in $y_0$.

Since $\tau$ is an invertible linear transformation, it is actually an automorphism of $L(R)$ and, 
in particular, the zeroness, simplicity and minimality properties of $C$ are invariant under it.
\end{proof}

An {\em ideal} $I$ of $R$
with generators $f_i, i\in[m]$, is the set $\{\sum_{i\in[m]}q_if_i|q_i\text{'s}\in R\}$
and is denoted by the notation $\ideal{f_1,\ldots,f_m}$. For any $f\in R$, the three notations
$f\equiv0 (\mod I)$, $f\equiv0 (\mod f_1,\ldots,f_m)$ and $f\in I$, mean the same.

An $f\in R$ is called a {\em zerodivisor} of an ideal $I$ (or mod $I$) if $f\notin I$
and there exists a $g\in R\setminus I$ such that $fg\in I$. 

Let $u, v\in R$. It is easy to see that if $u$ is nonzero mod $I$ and is a {\em non}-zerodivisor 
mod $I$ then: $uv\in I$ iff $v\in I$. This can be seen as some sort of a ``cancellation
rule'' for non-zerodivisors. We show such a cancellation rule in the case of ideals arising
in $\sps$ circuits.

\begin{lemma}[Non-zerodivisor]\label{lem-non-zd}
Let $f_1,\ldots,f_m$ be multiplication terms generating an ideal $I$, let $\ell\in L(R)$ and
$g\in R$.
If $\ell\notin\radsp(I)$ then: $\ell g\in I$ iff $g\in I$.
\end{lemma}
\begin{proof}
Assume $\ell\notin\radsp(I)$. If $I=\{0\}$ then the lemma is of course true. So let us
assume that $I\ne\{0\}$ and $\rk(\radsp(I))=:r\in[n-1]$. As $\ell\notin\radsp(I)$ there exists
an invertible linear transformation $\tau:L(R)\rightarrow L(R)$ that maps each 
form of $\radsp(I)$ to
$\lrsp(x_1,\ldots,x_r)$ and maps $\ell$ to $x_n$. Now suppose that $\ell g\in I$. 
This means that there are $q_1,\ldots,q_m\in R$ such that 
$\ell g =$ $\sum_{i=1}^m q_if_i$. Apply $\tau$ on this to get:
\begin{equation}\label{eqn-non-zd}
x_ng'=\sum_{i=1}^m q_i'\tau(f_i). 
\end{equation}
We know that $\tau(f_i)$'s are free of $x_n$. Express $g', q_i'$-s as polynomials wrt $x_n$, say 
\begin{align}
g'   &= \sum_{j\geq0}a_jx_n^j, \text{ where } a_j\in\FF[x_1,\ldots,x_{n-1}] \\  
q_i' &= \sum_{j\geq0}b_{i,j}x_n^j, \text{ where } b_{i,j}\in\FF[x_1,\ldots,x_{n-1}]
\end{align}
Now for some $d\geq1$ compare the coefficients of $x_n^d$ on both sides of Equation 
(\ref{eqn-non-zd}). We get $a_{d-1}=\sum_{i=1}^m b_{i,d}\tau(f_i)$, thus 
$a_{d-1}$ and $a_{d-1}x_n^{d-1}$ are in $\ideal{\tau(f_1),\ldots,\tau(f_m)}$. Doing this
for all $d\geq1$, we get $g'\in\ideal{\tau(f_1),\ldots,\tau(f_m)}$, hence 
$g=\tau^{-1}(g')\in\ideal{f_1,\ldots,f_m}=I$. This finishes the proof.
\end{proof}

All the ideals arising in this work are {\em homogeneous}, i.e. their generators are 
homogeneous polynomials. These ideals have some nice properties, as shown below.
Degree $\degr(\cdot)$ refers to the total degree unless there is a subscript specifying
the variable as well. 

\begin{lemma}[Homogeneous ideals]\label{lem-homo-ideal}
Say, $f_1,\ldots,f_m,g$ are homogeneous polynomials in $R$. Then,
 
1) If $\degr(g)<\degr(f_m)$ then: $g\in\ideal{f_1,\ldots,f_m}$ iff $g\in\ideal{f_1,\ldots,f_{m-1}}$.  

2) If $\degr(g)=\degr(f_m)$ then: $g\in\ideal{f_1,\ldots,f_m}$ iff $\exists a\in\FF$, 
$(g+af_m)\in\ideal{f_1,\ldots,f_{m-1}}$.
\end{lemma}
\begin{proof}
Say, $g\in\ideal{f_1,\ldots,f_m}$. Then, by definition, there exist $q$'s in $R$ such that,
\begin{equation}\label{eqn-homo-ideal}
g=\sum_{i=1}^m q_if_i.
\end{equation}
Let $d:=\degr(g)$.
If we compare the monomials of degree $d$ on both sides of Equation (\ref{eqn-homo-ideal})
then the LHS gives $g$. In the RHS we see that an $f_i$ of degree $d_i$ contributes $[q_i]_{(d-d_i)}f_i$,
where $[q]_j$ is defined to be the sum of the degree $j$ terms of $q$ (and, zero if $j<0$). Thus,
$g=\sum_{i=1}^m [q_i]_{d-d_i}f_i$. This equation proves both the properties at once.
\end{proof}

We show below that a congruence of two multiplication terms modulo an ideal, generated by terms, leads to
a matching via the radical-span. 
\begin{definition}[$L_U(\cdot), L^c_U(\cdot)$]
For a multiplication term $f$ and a subspace $U\subseteq L(R)$ define
$L_U(f):=L(f)\cap U$ and $L^c_U(f):=L(f)\setminus U$.
\end{definition}

\begin{lemma}[Congruence to Matching]\label{lem-I-match}
Let $I$ be an ideal generated by multiplication terms $\{f_1,\ldots,f_m\}$ and define $U:=\radsp(I)$. 
Let $g,h$ be multiplication terms
such that $g\equiv h\not\equiv0\ (\mod I)$. Then there is a $U$-matching between 
$L_U(g), L_U(h)$ and one between $L^c_U(g), L^c_U(h)$.
\end{lemma}
\begin{proof}
Define $g_0:=M(L_U(g))$ and $h_0:=M(L_U(h))$. Suppose the list $L_U(g)$ is larger than the 
list $L_U(h)$. By the congruence we have $h\in\ideal{I,g_0}$. As $\radsp(I,g_0)=U$,
by Lemma \ref{lem-non-zd} we can drop the non-$U$ forms of $h$ to get $h_0\in\ideal{I,g_0}$.
As $\ideal{I,g_0}$ is a homogeneous ideal and $\degr(h_0)<\degr(g_0)$ we get by Lemma 
\ref{lem-homo-ideal} that $h_0\in I$. But this means $h\in I$, which contradicts the 
hypothesis. Thus, $\degr(h_0)\geq\degr(g_0)$ and by symmetry we get them infact equal.
Thus, the lists $L_U(g), L_U(h)$ are of equal size, which trivially $U$-matches
them. 

We will show that for any $\ell\in L(R)\setminus U$, the number of 
forms that are similar to $\ell$ mod $U$ in $L^c_U(g)$ is equal to that in $L^c_U(h)$.
This fact will prove the lemma as it shows that every form in $L^c_U(g)$ can be 
$U$-matched to a distinct form in $L^c_U(h)$. 

Pick an $\ell\in L(R)\setminus U$. Let $g_1$ be the product of the forms that are similar to $\ell$ 
mod $U$ in $L^c_U(g)$ (if none exist then set $g_1=1$), similarly define $h_1$ from $h$. 
Suppose $\degr(h_1)<\degr(g_1)=:d$.
By the congruence we have $h\in\ideal{I,g_1}$. As $\radsp(I,g_1)=U\oplus\FF\ell$,
by Lemma \ref{lem-non-zd} we can drop the non $\lrsp(U,\ell)$ forms of $h$ to get 
\begin{equation}\label{eqn-I-match}
h_0h_1\in\ideal{I,g_1}.
\end{equation} 
Define $r:=\rk(U)$ which has to be $>0$, as otherwise $I=\ideal{1}$ contradicting $h\notin I$. 
Pick an invertible linear transformation $\tau:L(R)\rightarrow L(R)$ such that forms in $U$ are mapped
inside $\lrsp(x_1,\ldots,x_r)$ and $\ell\mapsto x_n$. Apply $\tau$ on Equation (\ref{eqn-I-match}) to 
get 
\begin{equation}\label{eqn-I-match-2}
h_0'h_1'=\sum_{i=1}^m q_if_i'+qg_1', 
\end{equation}
where $h_0'$ and $f_i'$-s are in $\FF[x_1,\ldots,x_r]$; $h_1'$ is $1$ or is a polynomial with 
$\degr_{x_n}\in[d-1]$; 
$g_1'$ is a polynomial with $\degr_{x_n}=d$; and $q\text{'s}\in R$. With these conditions if we
compare the coefficients of $x_n^d$ on both sides of Equation (\ref{eqn-I-match-2}) then we 
get $q\in\ideal{f_1',\ldots,f_m'}$, hence $\tau^{-1}(q)\in\ideal{f_1,\ldots,f_m}=I$. Thus, applying
$\tau^{-1}$ on Equation (\ref{eqn-I-match-2}) we get $h_0h_1\in I$, so $h\in I$, contradicting the 
hypothesis. Thus, 
$\degr(h_1)\geq\degr(g_1)$ and by symmetry we get them infact equal. This shows the number of 
forms that are similar to $\ell$ mod $U$ in $L^c_U(g)$ is equal to that in $L^c_U(h)$, finishing the 
proof.
\end{proof}

One pleasant consequence of $K$-matching all the multiplication terms in an identity is that we get
a smaller identity, using linear forms solely from $K$, called the {\em nucleus identity}. To see that we 
use a metric associated with matchings, first introduced in \cite{SS09}.

\begin{definition}[Scaling factor]
Let $K$ be a subspace of $L(R)$ and $L_1, L_2$ be two lists of linear forms in $L(R)\setminus K$. 
Let $\pi$ be a $K$-matching between $L_1, L_2$. Then for every $\ell\in L_1$, there is a {\em unique}
$c_\ell\in\FF^*$ such that $\pi(\ell)\in c_\ell\ell+K$ (if there is another $d\in\FF$ with 
$\pi(\ell)\in d\ell+K$, then $(c_\ell-d)\ell\in K$, implying $\ell\in K$, a contradiction). 

We define the {\em scaling factor of $\pi$}, $\scal(\pi):=\prod_{\ell\in L_1}c_\ell$.  
\end{definition}

\begin{lemma}[Nucleus identity]\label{lem-nucleus}
Suppose $C=\sum_iT_i$ is a $\sps(k,d)$ identity and $K$ is a subspace of $L(R)$ such that 
$T_1, T_i$ are $K$-matched, for all $i\in[k]$. Then the terms 
$M(L_K(T_i))$, for $i\in[k]$, are all of the same degree, say $d'$, and form a $\sps(k,d')$ identity 
$\sum_{i\in[k]}\alpha_i M(L_K(T_i))$, for some $\alpha_i\in\FF^*$.
\end{lemma}
\begin{proof}
Since $T_1, T_i$ are $K$-matched, we get from the definition of matching that terms $M(L_K(T_1))$,
$M(L_K(T_i))$ have the same degree $d'\geq0$. Furthermore, $M(L^c_K(T_1))$
and $M(L^c_K(T_i))$ are also $K$-matched, call this induced
matching $\pi_i$. As all the forms in $L^c_K(T_1)$ are outside $K$, the scaling factor $\scal(\pi_i)$ is
well defined, for all $i\in[k]$. 

Fix a subspace $U$ such that $L(R)=K\oplus U$ and let $r:=\rk(K)$. Fix an invertible linear 
transformation $\tau:L(R)\rightarrow L(R)$ that maps $K$ to $\lrsp(x_1,\ldots,x_r)$. It follows that
for any form $\ell\in L^c_K(T_1)$, $\tau(\ell)$ is a form with a nonzero coefficient wrt some $x_i$, 
$i>r$ (otherwise $\tau(\ell)\in\lrsp(x_1,\ldots,x_r)$, thus $\ell\in K$, a contradiction). Call the
largest such $i$, $j_\ell$. If we look at the product (note: it is over a list so it could have repeated factors), 
\begin{equation}
\alpha_1:=\prod_{\ell\in L^c_K(T_1)}[x_{j_\ell}]\tau(\ell)
\end{equation}
($[\vec{x}^{\vec{i}}]f$ gives the coefficient of the monomial $\vec{x}^{\vec{i}}$ in $f$), then
it is the coefficient of $\prod_{\ell\in L^c_K(T_1)}x_{j_\ell}$ in $\tau(M(L^c_K(T_1)))$, in other words,
$\alpha_1$ is its leading coefficient wrt lexicographic  ordering of variables. Note that, for $i\in[k]$, $\pi_i$ 
still $\tau(K)$-matches $\tau(L^c_K(T_1)), \tau(L^c_K(T_i))$ with the same scaling factor (if 
$\pi_i(\ell)\in c_\ell\ell+K$ then $\tau(\pi_i(\ell))\in c_\ell\tau(\ell)+\tau(K)$). This means that the leading 
coefficient of $\tau(M(L^c_K(T_i)))$ is $\scal(\pi_i)\cdot\alpha_1=:\alpha_i$, for all $i>1$. Thus, we have 
pinpointed the coefficient of $\prod_{\ell\in L^c_K(T_1)}x_{j_\ell}$ in $\tau(M(L^c_K(T_i)))$ as 
$\alpha_i$, for all $i\in[k]$. Now compare the coefficients of $\prod_{\ell\in L^c_K(T_1)}x_{j_\ell}$ in 
the identity $\tau(C)=0$. This gives $\sum_{i\in[k]}\alpha_i\cdot\tau(M(L_K(T_i)))$ $=0$. Applying the 
inverse of $\tau$, we get the nucleus identity.
\end{proof}

In Lemma \ref{lem-non-zd} we have already come across a cancellation rule for non-zerodivisors. 
Here we see a situation in which it is stronger.

\begin{lemma}[Cancellation]\label{lem-cancel}
Let $K$ be some subspace of $L(R)$ and let $\ell_1,\ldots,\ell_m\in L(R)\setminus K$ be
linearly independent modulo $K$. Let $f_1,\ldots,f_m$ be multiplication terms similar to powers of
$\ell_1,\ldots,\ell_m$ respectively modulo $K$ (i.e. each form in $f_i$ is in $(\FF^*\ell_i+K)$). 
Let $\ell\in L(R)^*$ such that for some $s\in[m]$, $\ell\in\FF\ell_s+K$. Then,  
for any polynomial $f\in R$, 
$$\ell f\in\ideal{f_1,\ldots,f_m} \text{ iff } f\in\ideal{f_1,\ldots,\frac{f_s}{\gcd(f_s,\ell)},\ldots,f_m}.$$
\end{lemma}
\begin{proof}
Suppose $\ell f\in\ideal{f_1,\ldots,f_m}$. Then, by definition, there exist $q$'s in $R$ such that,
\begin{equation}\label{eqn-cancel}
\ell f=\sum_{i=1}^m q_if_i.
\end{equation}
Additionally assume these $q_i$-s to be such that the set 
$J:=\{j\in[m]\setminus\{s\}\ |\ \ell\nmid q_j\}$ is the smallest possible.
If $\ell|q_i$, for all $i\in[m]\setminus\{s\}$, then $\ell$ has to divide $q_sf_s$. This means that
$\ell$ has to divide $q_s\gcd(\ell,f_s)$, thus we get, 
$$f=\sum_{i\in[m]\setminus\{s\}}\frac{q_i}{\ell}f_i\ +\ 
\frac{q_s\gcd(\ell,f_s)}{\ell}\cdot\frac{f_s}{\gcd(f_s,\ell)}$$ 
and we are done. 

So the remaining case is when the set $J:=$ $\{j\in[m]\setminus\{s\}\ |\ \ell\nmid q_j\}$ is nonempty.
Fix an element $j^*\in J$. Consider ideal $I:=\ideal{\{\ell,f_s\}\cup \{f_j|j^*\ne j\in J\}}$. Reducing
Equation (\ref{eqn-cancel}) modulo $I$ we get, $q_{j^*}f_{j^*}\equiv0 (\mod I)$. Note that 
$\radsp(I)\subseteq K+\lrsp(\{\ell_j|j^*\ne j\in[m]\})$ while each form in $L(f_{j^*})$ is in
$(\FF^*\ell_{j^*}+K)$ disjoint from $\radsp(I)$, thus by Lemma \ref{lem-non-zd} we can drop $f_{j^*}$ 
from the last congruence and get $q_{j^*}\in I$. This means 
$q_{j^*}f_{j^*}\in\ideal{\{\ell f_{j^*},f_s\}\cup \{f_j|j^*\ne j\in J\}}$. We plug this in the
$j^*$-th summand of Equation (\ref{eqn-cancel}) and after simplifications get 
(verify that the $[m]\setminus(\{s\}\cup J)$ summands are unaffected):
\begin{align*}\label{eqn-cancel-2}
\ell f &= \sum_{i=1}^m q_if_i \\
&= q_s'f_s + (\ell q_{j^*}')f_{j^*} + \sum_{j\in J\setminus\{j^*\}}q_j'f_j +
\sum_{j\in[m]\setminus(\{s\}\cup J)}q_jf_j
\end{align*}
Notice that for $j\in[m]\setminus(\{s\}\cup J)$, $\ell$ divides $q_j$, thus the above equation contradicts
the assumed minimality of $J$. This shows that $J$ was empty to begin with, thus finishing the proof. 
\end{proof}

%
%
%

\end{document}